\documentclass[reqno]{amsart}

\usepackage{amssymb,amsmath,amscd,amsthm}

\usepackage{mathrsfs}
\usepackage{latexsym}
\usepackage{graphicx}
\usepackage{dsfont}

\newtheoremstyle{myplain}{}{}{\it}
{0pt}{\scshape}{}{ }{\thmname{#1}\thmnumber{ #2}\thmnote{ (#3)}}
\newtheoremstyle{mydefinition}{}{}{}
{0pt}{\scshape}{}{ }{\thmname{#1}\thmnumber{ #2}\thmnote{ (#3)}}

\theoremstyle{myplain}

    \newtheorem{Def}{Definition}[section]
        \newtheorem{Lem}[Def]{Lemma}

        \newtheorem{theo}[Def]{Theorem}
        \newtheorem{prop}[Def]{Proposition}
        \newtheorem{rem}[Def]{Remark}
        \newtheorem{hyp}[Def]{Hypothesis}
        \newtheorem{cor}[Def]{Corollary}

\newcommand{\natop}[2]{\genfrac{}{}{0pt}{}{#1}{#2}}
\newcommand{\expord}[1]{O\left(e^{-\frac{#1}{\hbar}}\right)}

\DeclareMathOperator{\dist}{dist}
\DeclareMathOperator{\Span}{span}

\DeclareMathOperator{\supp}{supp}

\DeclareMathOperator{\spec}{spec}

\DeclareMathOperator{\diag}{diag}

\DeclareMathOperator{\dvol}{d\,vol}
\DeclareMathOperator{\Id}{id}
\DeclareMathOperator{\grad}{grad}

\newcommand{\id}{\mathds{1}}

\numberwithin{equation}{section}

\newcommand{\beqa}{\begin{eqnarray*}}
\newcommand{\eeqa}{\end{eqnarray*}}
\renewcommand{\hat}{\widehat}
\newcommand{\bauf}{\begin{itemize}}
\newcommand{\eauf}{\end{itemize}}
\newcommand{\ben}{\begin{enumerate}}
\newcommand{\een}{\end{enumerate}}

\renewcommand{\o}{\omega}
\renewcommand{\O}{\Omega}
\newcommand{\ep}{\varepsilon}

\newcommand{\R}{{\mathbb R} }

\newcommand{\C}{{\mathbb C}}
\newcommand{\N}{{\mathbb N}}

\newcommand{\hZ}{\frac{\mathbb Z}{2}}

\newcommand{\Ce}{\mathscr C}

\newcommand{\De}{\mathscr D}
\newcommand{\Ge}{\mathscr G}
\newcommand{\E}{{\mathcal E}}
\newcommand{\F}{{\mathcal F}}

\newcommand{\cL}{{\mathcal L}}

\newcommand{\End}{\mathrm{End}}
\newcommand{\Eh}{\mathscr{E}}
\newcommand{\dd}{\mathrm{d}}
\newcommand{\<}{\left\langle\!\!\left\langle}
\renewcommand{\>}{\right\rangle\!\!\right\rangle}
\newcommand{\scl}{\langle\!\langle}
\newcommand{\scr}{\rangle\!\rangle}

\title{The Tunneling Effect for Schr\"odinger operators on a vector bundle}

\author{Markus Klein \and Elke Rosenberger*}

\address{
Universit\"at Potsdam\\ Institut f\"ur
Mathematik \\ Karl-Liebknecht-Str. 24-25\\ 14476 Potsdam}
\email{mklein@math.uni-potsdam.de, erosen@uni-potsdam.de}

\date{\today}


\begin{document}

  \begin{abstract}
  In the semiclassical limit $\hbar\to 0$, we analyze a class of self-adjoint Schr\"odinger operators $H_\hbar = \hbar^2 L + \hbar W + V\cdot \mathrm{id}_\Eh$ acting on sections 
of a vector 
bundle $\Eh$ over an oriented Riemannian manifold $M$ where $L$ is a Laplace type operator, $W$ is an endomorphism field
  and the potential energy $V$ has non-degenerate minima at a finite number of points $m^1,\ldots m^r \in M$, called potential wells.
  Using quasimodes of WKB-type near $m^j$ 
  for eigenfunctions associated with the low lying eigenvalues of $H_\hbar$, we analyze the tunneling effect, i.e. the splitting between low lying eigenvalues, which e.g. arises in certain symmetric configurations. 
  Technically, we treat the coupling between different potential wells by an interaction matrix and we consider the case of a single minimal geodesic 
  (with respect to the associated Agmon metric) connecting two potential wells 
  and the case of a submanifold of minimal geodesics of dimension $\ell + 1$. 
  This dimension $\ell$ determines the polynomial prefactor for exponentially small eigenvalue splitting.
  \end{abstract}
  
  \maketitle 

\section{Introduction}

In this paper, we study the low lying spectrum of a
Schr\"odinger operator $H_\hbar$ on a vector bundle $\Eh$ over a smooth oriented Riemannian manifold $M$. More precisely, in the
limit $\hbar \to 0$, we analyze the 
tunneling effect for operators of the form
\[
   H_\hbar = \hbar^2 L + \hbar W + V \cdot \Id_\Eh
\]
acting on the space $\Gamma^\infty(M, \Eh)$ of smooth sections of $\Eh$, where $L$ is a symmetric Laplace type operator (i.e. a
second order differential operator on $\Eh$ with principal symbol $\sigma_L(x,\xi) = |\xi|^2$), 
$W\in \Gamma^\infty(M, \End(\Eh))$ is a smooth symmetric endomorphism field over $M$ and the potential energy $V\in \Ce^\infty(M, \R)$
has a finite number of non-degenerate minima $m^1, \ldots m^r$.

Operators of this type arise e.g.\ in Witten's perturbation of the de Rham complex, where $H_\hbar$ is the square of 
the Dirac type operator
\[
 Q_\phi= \hbar \,  \bigl( \dd_\phi + \dd^*_\phi  \bigr) ,   \qquad \dd_\phi= e^{-\phi/\hbar} \dd  e^{\phi/\hbar}
=\dd + \dd\phi \wedge, 
\]
where $\phi$ is a Morse function.
This operator acts on (the sections of) $\Eh=\Lambda^p M$ taking values in $\bigoplus  \Lambda^p M$, while its square
$P=Q_\phi^2$
maps the sections of $\Eh$ into itself. More explicitly, it is given by
\[
 P= \hbar^2 (\dd \dd^* + \dd^* \dd) + \hbar (\cL_{\grad \phi} + \cL_{\grad \phi}^*) + \|\dd \phi \|^2 \Id_\Eh,
 \]
 where $\cL_X$ denotes the Lie derivative in the direction of the vector field $X$, and $\grad \phi$ is the gradient of $\phi$ with respect  to the Riemannian metric $g$. Then 
$W:= \cL_{\grad \phi} + \cL_{\grad \phi}^*$ actually is $\Ce^\infty(M)$ linear and thus an endomorphism field as described above.
In this particular case, the endomorphism $W$ is non-vanishing, which is the reason for us to include this (somewhat unusual) term in our considerations. 
This operator has been considered in detail in \cite{helffer-sjostrand-4}, giving much mathematical detail to the original paper  \cite{witten};  see also \cite{hkn}, which derives full asymptotic expansions of all low-lying eigenvalues 
using an inductive approach which builds on the results of \cite{begk}, \cite{bgk}, \cite{eck} for generators of reversible diffusion operators, using a 
potential theoretic approach based on estimating capacities. It seems to be open if the latter approach could also be applied to operators as considered in this 
paper, and it is also open if {\em all} of the low-lying spectrum of operators as considered in this paper could be analysed by methods close to \cite{hkn}. Here the Witten-Laplacian seems to be special.

We use semi-classical quasimodes of WKB-type constructed in \cite{ludewig}, which are an important step in discussing tunneling problems, i.e.\
exponentially small splitting of eigenvalues for a self-adjoint realization of $H_\hbar$. In the scalar case, for $\dim M>1$, 
rigorous results in this field start with the seminal paper \cite{helffer-sjostrand-1} (for $M=\R^n$ or
$M$ compact).

In everything what follows, let $(M, g)$ be a (smooth) oriented $n$-dimensional Riemannian manifold and let $\pi:\Eh \rightarrow M$ be a complex vector bundle over 
$M$ equipped with an inner 
product $\gamma$ (i.e. a positive definite Hermitian form). Let $\mathrm{rk}\Eh$ denote the dimension of a fibre of $\Eh$. 
Denoting by $\dvol$ the Radon measure on $M$ induced from the Riemannian metric $g$, 
this allows to introduce the Hilbert space $L^2(M,\Eh)$ of (equivalence classes of) sections in $\Eh$ as the completion of $\Gamma_c^\infty(M, \Eh)$, the space of compactly supported 
smooth sections of $\Eh$\footnote{For $K\subset M$, we write $\Gamma^\infty_c(K, \Eh)$ to denote the sections of $\Eh$ compactly supported in $K$.}, with inner product
\begin{equation}\label{standard_skalar}
  \langle\!\langle u, v\rangle\!\rangle_\Eh = \int_M \gamma_m[ u, v ]\, \dvol (m) \quad\text{ and associated norm }\quad  \|u\|_\Eh^2=\langle\!\langle u,u\rangle\!\rangle_\Eh\; .
\end{equation}

Recall that a differential operator $L$ acting on sections of $\Eh$ is said to be of {\em Laplace type} if, in local coordinates $x$, it has the form
\begin{equation}\label{Laplace-type}
  L = - \Id_{\Eh} \sum_{i,j} g^{ij}(x)\frac{\partial^2}{\partial x^i\partial x^j} + \sum_{j} 
b_j \frac{\partial}{\partial x^j} + c
\end{equation}
where $\bigl(g^{ij}(x)\bigr)$ is the inverse matrix of the metric $\bigl(g_{ij}(x)\bigr)$ and $b_j, c\in \Gamma^\infty (M, \End (\Eh))$ are
endomorphism fields. Examples for Laplace-type operators are the Hodge-Laplacian
$d d^* + d^*d$ on 
$p$-forms (in particular, for $p=0$, this is the Laplace-Beltrami operator) and the square of a generalized Dirac operator acting on spinors.  Also, second order elliptic operators $L$ in 
divergence form, i.e, $L=\sum_{i,j} \partial_i a^{ij} \partial_j$,
on open subsets of $\R^n$  belong to this class. To the best of our knowledge, even for scalar operators of this form the tunneling effect has not been treated explicitly in the literature. 
The geometrically formulated theorems of our paper cover in particular this special case, closing an obvious gap in the literature.

Moreover, for any Laplace type operator $L$ on $\Eh$ which is symmetric on $\Gamma^\infty_c(M, \Eh)$  with respect to the inner product
$\langle\!\langle \,.\,, \,.\,\rangle\!\rangle_{\Eh}$, 
there exists a unique metric connection 
$\nabla^\Eh: \Gamma^\infty(M, \Eh) \rightarrow \Gamma^\infty(M, T^*M\otimes \Eh)$ on $\Eh$ and a symmetric endomorphism field 
$U \in \Gamma^\infty(M, \End_{sym}(\Eh))$ such that
\begin{equation}\label{L-Darstellung}
  L = (\nabla^\Eh)^* \nabla^\Eh + U
\end{equation}
(see \cite{ludewig}, Remark 2.1).
In the following we always assume $\nabla^\Eh$ to be the metric connection such that \eqref{L-Darstellung} holds.\\
We denote by $\langle\,.\,,\,.\,\rangle_m$ and $|\,.\,|_m$ the inner product and norm on $T^*_mM$ induced by $g$ (we feel free to sometimes suppress the index $m \in M$) and we use the same symbols for the extension of the inner product and the norm to the complexified cotangent bundle $T^*_\C M =T^*M \otimes \C$. Then  $T^*_\C M\otimes \Eh$ is well defined as a bundle (since fibrewise both factors are complex vector spaces), and we denote by 
$\langle\!\langle \,.\,, \,.\,\rangle\!\rangle_\otimes$ and $\|\,.\,\|_\otimes$ the inner product and norm on $L^2(M, T^*_\C M\otimes \Eh)$; both are induced from the inner product on the fibres of $T^*_\C M\otimes \Eh$ which for complex one-forms 
$\alpha, \beta$ and sections $u,v$ of $\Eh$ is given by
$$\langle \alpha \otimes u,\beta \otimes v \rangle_\otimes= \langle \alpha, \beta \rangle_1 \gamma [u,v]$$
and then extends to the full tensor product by linearity (see also the beginning of Section 4 for some more standard details on the inner product $\langle \cdot, \cdot \rangle_p$ in the fibres of the complexified exterior bundle $\Lambda_\C^p(M)$). We feel free to sometimes suppress the subscript $\C$. As a general rule, all of our inner products are antilinear in the first and linear in the second factor.

Our setup is the following.

\begin{hyp}\label{setup1}
For $\hbar>0$, we consider Schr\"odinger operators $H_\hbar$ acting on $L^2 (M, \Eh)$
of the form
\begin{equation}\label{schrodinger}
  H_\hbar = \hbar^2 L + \hbar W + V \cdot \Id_\Eh
\end{equation}
where $L$ is a symmetric Laplace type operator as above, 
$W \in \Gamma^\infty(M, \End_{sym}(\Eh))$ is a symmetric endomorphism field and $V \in C^\infty(M, \R)$. Furthermore, we shall assume:
\ben
\item The potential $V$ 
is non-negative and there is a compact subset $K\subset M$ and $\delta>0$ such that $V(m)\geq \delta$ for all $m\in M\setminus K$.
\item $V$ has non-degenerate minima at a finite number of points $m^j \in M,\, j\in\{1, \ldots, r\}:= \mathcal{C}$, i.e., $V(m^j) = 0$ and $\nabla^2V|_{m^j}>0$.
\item If $U$ is the endomorphism field given in equation \eqref{L-Darstellung} with respect to $L$ in \eqref{schrodinger}, the symmetric 
endomorphism field $\hbar U + W$ is bounded from below, i.e. there is a positive constant $C<\infty$ such that
\begin{equation}\label{1_hyp1}
 \scl u, (\hbar U + W)u\scr_\Eh \geq - C \|u\|_\Eh^2\, , \qquad u\in \Gamma^\infty_c(M, \Eh)\, ,
\end{equation}
uniformly for $\hbar\in (0,1]$.
\een
\end{hyp}

It is then clear that (for $\hbar \in (0,1]$) the operator $H_\hbar$ with domain $\Gamma^\infty_c(M, \Eh)$ is a semi-bounded, symmetric and densely defined operator in the Hilbert space 
$L^2(M,\Eh)$. Thus, by a well known result of abstract spectral theory, its associated semi-bounded quadratic form
\begin{equation}\label{quadratic-form}
q_\hbar(u):= \scl u, H_\hbar u\scr_\Eh \, , \qquad u\in \Gamma^\infty_c(M, \Eh)
 \end{equation}
is closable. Passing to the closure of $q_\hbar$ and using the representation theorem for symmetric, semi-bounded, closed forms yields a distinguished self-adjoint operator, the 
Friedrich's extension of $H_\hbar: \Gamma^\infty_c(M, \Eh) \rightarrow L^2(M, \Eh)$ (which by usual abuse of notation we shall also denote by $H_\hbar$).\\
We recall that, if $M$ is assumed to be complete, $H_\hbar: \Gamma^\infty_c(M, \Eh) \rightarrow L^2(M, \Eh)$ is actually essentially self-adjoint\footnote{See \cite{chernoff} and 
\cite{strichartz} for proofs using finite propagation speed and nice partitions of unity, respectively. Both papers do not formally cover precisely the class of operators of Laplace
type considered in the present paper, but both methods generalize to our class of operators on bundles (e.g., the propagation speed in \cite{chernoff} only depends 
on the principal symbol of $H_\hbar$ for fixed $\hbar$, thus being independent of $U$ and $W$).}.
Furthermore, if $M$ in addition is assumed to be of bounded geometry, various different natural approaches to the definition of Sobolev spaces for $\Eh$ all lead to identical results
(see \cite{aubin} and \cite{eichhorn}).
For the purpose of this paper, none of this seems to be relevant. We shall stick to the Friedrich's extension of $H_\hbar$.\\
Similarly, for any open $\Omega\subset M$ with compact closure in $M$ we shall define the Dirichlet realization $H_\hbar^{\overline{\Omega}}$ of $H_\hbar$ in $\Omega$ by
Friedrich's extension of $H_\hbar: \Gamma_c^\infty (\Omega, \Eh)\rightarrow L^2(\Omega, \Eh)$.

We remark that the operator $H_\hbar$ given in \eqref{schrodinger} is not necessarily real, i.e.\ it does not commute with complex conjugation, in contradistinction to the  more special case of the Witten Laplacian $P$ introduced above. Thus, in the general case, the well known  Beurling-Deny criteria do not apply, and even the groundstate of our operator $H_\hbar$ may well be degenerate. It is thus natural to treat tunneling in this degenerate setting and we shall do so in due course.

However, under
semi-classical quantization ($\xi \mapsto -i\hbar \dd$ in some reasonable sense) its principal $\hbar$-symbol 
\begin{equation}\label{hauptsymbol}
 \sigma_H: T^*M \rightarrow \End (\Eh)\, , \quad \sigma_H(m, \xi) = \bigl( |\xi|^2 + V(m)\bigr) \Id_\Eh, \quad (m,\xi) \in T^*_mM
\end{equation}
is both real and scalar. This is crucial for our construction.
Thus our assumptions exclude Schr\"odinger operators with magnetic field (the operator $(i\hbar \dd + \alpha)^*(i\hbar \dd + \alpha)$, with a $1$-form 
$\alpha$ describing the magnetic potential, has non-real
principal $\hbar$-symbol, see e.g.\ \cite{helffer-kondryukov-1}) or with endomorphism valued potential $V$ as needed e.g.\ for molecular Hamiltonians in the
Born-Oppenheimer approximation (see e.g.\ \cite{klein-seiler}).

Defining the hyperregular Hamiltonian
\begin{equation}\label{tildehnull}
\tilde{h}_0 : T^*M \rightarrow \R\, , \quad \tilde{h}_0(m, \xi) = |\xi|^2 - V(m) 
\end{equation}
one has $\sigma_H (m,\xi) = -\tilde{h}_0(m, i\xi) \Id_\Eh$, and one can use the theory developed in \cite{kleinro} (or results given in 
\cite{helffer-sjostrand-1}) to introduce an adapted geodesic distance on $M$. There it is shown that the (Agmon)-distance $d$ on $M$ given by the Agmon-metric $\dd s^2 = V g$ (which is the 
Jacobi metric of classical mechanics on the Riemannian manifold $(M, g)$ for the Hamiltonian function $|\xi|^2 + V$ at energy zero) is Lipschitz everywhere and
smooth near the potential wells $m^j,\, j\in\mathcal{C}$. 

Moreover, defining $d^j(m):= d(m, m^j)$ for $j\in\mathcal{C}$,
by \cite{kleinro}, Theorem 1.6 (or see\\ \cite{helffer-sjostrand-1}, Proposition 3.1) there are neighborhoods $\Omega_j$ of $m^j$ such that
\begin{align}\label{eikonal}
\tilde{h}_0(m, \dd d^j(m)) &= 0 \, , \qquad m\in \Omega_j\\
\tilde{h}_0(m, \dd d^j(m)) &\leq 0 \, , \qquad m\in M \; .\label{eikonalin}
\end{align}
These are the eikonal equation and  the eikonal inequality.

The central result of this paper is a very precise asymptotic formula for the splitting of certain low-lying eigenvalues of the operator $H_h$. 

As we shall recall below, power series expansions of low lying eigenvalues are in leading order given by the 
harmonic approximation and can be derived by a certain perturbation theory. 

Furthermore, in situations of symmetry or near symmetry certain eigenvalues, 
including but not restricted to the groundstate, are almost degenerate in the sense of being 
exponentially close. On the other hand, such almost degeneracy of eigenvalues can be considered as a spectral picture of (possibly almost) 
underlying symmetry in a geometrical sense. In such a situation, approximating the spectrum of $H_h$ by 
the spectrum of appropriately chosen operators with Dirichlet boundary conditions on certain subsets $\O_j$ of $M$ might give a truly degenerate spectrum even for the groundstate, and the true spectrum of $H_h$ can then 
perturbatively be recovered through an (exponentially small) so called interaction matrix. 

Crucial ingredients are the minimal geodesics 
between the potential wells and the distance between the wells in the Agmon metric $\dd s^2 = V g$.  

While even in the scalar case such an analysis of eigenvalue splitting is often restricted to the case of one unique minimal geodesic (and possibly the groundstate), we here analyse the more complex situation where these
minimal geodesics might form submanifolds of dimension $\ell +1$ (the above mentioned case of a unique geodesic is then a special case for $\ell=0$). 
See \cite{kr2,kr3} for somewhat similar results for a class of difference operators on a 
scaled lattice in $\R^n$. Intuitively (in both cases), larger eigenvalue splitting is connected to more tunneling (or larger conductance) between the wells and 
thus to the dimension of connecting geodesics as these provide 
(in some sense which we shall not even try to make precise) optimal tunneling paths for a quantum particle. For the operators considered in this paper, this is 
made precise in our Theorem 6.7 below which is our central result.
 
 The outline of the paper is as follows. In Section 2 we  justify the harmonic approximation in our setting and sketch the proof of the basic spectral stability result. 
 In Section 3 we prove Agmon-type estimates in the semiclassical limit 
 for the decay of eigenfunctions for the Dirichlet operators on certain sets $\O_j$ in $M$ containing only one potential well. 
 As usual, these are a crucial ingredient for subsequent WKB expansions. We emphasize that these estimates 
 only involve certain structural identities for our operator; in particular, computations with a full expansion of everything  in local coordinates are not required. 
 Section 4 introduces the crucial above mentioned interaction matrix, and 
 Section 5 analyzes it in more detail in important special cases. Section 6 is the heart of the paper, culminating in the above mentioned Theorem 6.7 which 
 contains a full asymptotic expansion of the interaction matrix (and thus of 
 the eigenvalue splitting, including but not restricted to the groundstate). As usual for such a precise result this requires additional geometric assumptions. 
 Most importantly, the outgoing manifolds (with starting point $m^j$) for the 
 Hamilton field of $\tilde{h}_0$ have to be parametrized as Lagrange manifolds by the geodesic distance $d^j(m^j, \cdot)$. 
 And it is here that we prove (under appropriate assumptions) that the constructed WKB-expansions are 
 actually very close to the true eigenfunctions of Dirichlet realizations of $H_h$ in $\O_j$ which justifies replacing the Dirichlet eigenfunctions in the interaction 
 matrix by these WKB functions. The proof of our main theorem 
 then requires  combining all this preparatory work with certain explicit calculations in local coordinates involving a form of the Morse Lemma with parameters and stationary phase.

\section{Harmonic Approximation}\label{sec2}

In this section we shall show that the lowest $N$ eigenvalues of $H_\hbar$ are given by the lowest $N$ eigenvalues of the direct sum of associated harmonic 
oscillators at the potential wells, up to an error $O(\hbar^{6/5})$ as $\hbar\to 0$. 

Denoting the fibre over $m\in M$ by $\Eh_{m}:= \pi^{-1}(m)$, we define the harmonic oscillator at $m^j, \, j=1, \ldots r,$ associated to $H_\hbar$ as the operator on $\Ce^\infty (T_{m^j}M, \Eh_{m^j})$ given by
\begin{equation}\label{harmoss}
 H_{m^j, \hbar} f(X) := \Bigl( \hbar ^2 \Delta_{T_{m^j}M} + \hbar W(m^j) + \frac{1}{2} \nabla^2 V|_{m^j} (X, X)\Bigr) f(X)\, , \quad X\in T_{m^j}M\, ,
\end{equation}
where $\Delta_{T_{m^j}M}$ denotes the Laplacian on $T_{m^j}M$ induced by the metric $g_{m^j}$ and $\nabla^2 V|_{m^j}$ denotes the Hessian of $V$ at $m^j$.

\begin{rem}
The spectrum of the harmonic oscillators $H_{m^j, \hbar},\, j\in  \mathcal{C},$ 
consists of the numbers
\begin{equation} \label{LocalEigenvalues}
  \hbar e^j_{\gamma, \ell} = \hbar \Bigl( \mu^j_\ell + \sum_{k=1}^n (2\gamma_k +1) \lambda^j_k  \Bigr)\, , \qquad \gamma \in \N^n, ~~ \ell=1, \dots, \mathrm{rk}\,\Eh
\end{equation}
where $\lambda^j_1, \dots, \lambda^j_n$ are the eigenvalues of $\frac{1}{2}\nabla^2 V|_{m^j}$ and $\mu^j_1, \dots, \mu^j_{\mathrm{rk} \Eh}$ are the eigenvalues 
of $W(m^j)$ (see e.g. \cite{cycon}, \cite[Section 8.10]{reed-simon-1}). \\
\end{rem}

As a preparation, assume $\chi_k$ are non-negative smooth functions with $\sum_k \chi_k^2 = 1$ (i.e. a quadratic partition of unity). Then a short calculation with double commutators gives
\[ (\nabla^\Eh)^*\nabla^\Eh = \sum_k \chi_k (\nabla^\Eh)^*\nabla^\Eh \chi_k - \sum_k |d\chi_k|^2\, , \]
yielding the identity
\begin{equation}\label{Hpartition}
 H_\hbar = \sum_k \chi_k H_\hbar \chi_k - \hbar^2 \sum_k |d\chi_k|^2\, .
\end{equation}
(called the IMS-Localization formula in \cite{cycon}).

\begin{Lem}\label{lemma_harm}
 Assume that $H_\hbar$ is of the form \eqref{schrodinger}, satisfying Hypothesis \ref{setup1}. Then there is $\hbar_0>0$ such that for all $0<\hbar \leq \hbar_0$
 \[  \inf \sigma_{ess} (H_\hbar) \geq \frac{\delta}{2}  \]
 where $\sigma_{ess}(H_\hbar)$ denotes the essential spectrum of $H_\hbar$.
\end{Lem}

\begin{proof}
 By Hypothesis \ref{setup1}, for $\hbar_0$ sufficiently small, there is a compact set $K\subset M$ such that for all $\hbar\leq \hbar_0$
 \begin{equation}\label{1_lemma_harm}
 \scl u, H_\hbar u\scr_\Eh \geq \frac{3}{4}\delta\, , \qquad u\in \Gamma_c^\infty (M\setminus K, \Eh)\, .
  \end{equation}
Now choose a smooth quadratic partition of unity $\chi_0^2 + \chi_1^2 = 1$, subordinate to the open cover $\Omega_0 = M\setminus K$ and $\Omega_1\supset K$, and a 
positive function $V_K\in C^\infty_0(M)$ such that $V_K \Id_\Eh + \hbar^2 U + \hbar W\geq \frac{3}{4}\delta$ on $\Omega_1$. Using \eqref{Hpartition} we
obtain for all $\hbar\leq \hbar_0$ and $u\in\Gamma_c^\infty(M, \Eh)$
\begin{equation}\label{2_lemma_harm}
 \scl u, (H_\hbar + V_K) u\scr_\Eh \geq \scl \chi_0 u, (H_\hbar + V_K) \chi_0 u\scr_\Eh + \scl \chi_1 u, (H_\hbar + V_K) \chi_1 u\scr_\Eh - O(\hbar^2 \|u\|_\Eh^2) \geq \frac{1}{2}\delta \|u\|_\Eh^2\, .
\end{equation}
Clearly, $V_K$ is relatively compact with respect to $H_\hbar$.
In fact, let us define the Sobolev space $H^1(M,\Eh)$ as the set of those $u\in L^2(M, \Eh)$ such that the distributional derivative $\nabla^\Eh u$ is 
in $L^2(M, T^*M\otimes \Eh)$, with norm
$||u||_{H^1} = ||u||_{\Eh} + ||\nabla^\Eh u||_{\otimes} $.  Since the domain $\De(H_\hbar)$ of the Friedrichs extension is contained in the form domain, one easily checks that 
$\De(H_\hbar)$ consists of functions in $H^1_{loc}(M,\Eh)$.  Furthermore, Rellich's compactness theorem holds in the following form: If $\Omega$ is an open 
set in $M$ with compact closure, the embedding $H^1(\Omega,\Eh) \to L^2(\Omega,\Eh)$ is compact\footnote{We remark that, since $M$ is neither complete nor of 
bounded geometry, our definition of the Sobolev space
$H^1(M,\Eh)$ need not necessarily coincide with the usual other natural definitions, see e.g. \cite{eichhorn}. For $H^1(\Omega,\Eh)$, 
however, all these ambiguities disappear by compactness of $\overline{\Omega}$.}. This gives, for $\lambda <0$  in the resolvent set, 
compactness of $V_K (H_\hbar - \lambda)^{-1}$.  
Thus Weyl's essential spectrum theorem gives in view of \eqref{2_lemma_harm}
\begin{equation}
\inf\sigma_{ess} H_\hbar = \inf \sigma_{ess} (H_\hbar + V_K) \geq \inf \sigma (H_\hbar + V_K) = \inf_{u\in\Gamma_c^\infty(M,\Eh), \|u\|_{\Eh}=1}\scl u, (H_\hbar + V_K)u\scr_\Eh \geq \frac{\delta}{2}\, ,
 \end{equation}
proving Lemma \ref{lemma_harm}.
\end{proof}

\begin{theo}\label{theo7}
 Assume that $H_\hbar$ satisfies Hypothesis \ref{setup1}. Denote by $\hbar e_\ell$ the $\ell$-th eigenvalue of $\bigoplus_{j=1}^r H_{m^j, \hbar}$, counting multiplicity. Then, for fixed $m\in\N$ and $\hbar$ sufficiently small, $H_\hbar$ has
 at least $m$ eigenvalues $E_\ell(\hbar), \ell=1, \ldots m,$ below $\inf \sigma_{ess}(H_\hbar)$ and 
 \begin{equation}\label{0_theo7}
   E_\ell(\hbar) = \hbar e_\ell + O(\hbar^{6/5}) \qquad (\hbar \to 0)\, . 
 \end{equation}
 \end{theo}

\begin{proof}
 
The proof follows closely the arguments of
\cite{simon-1}, \cite{simon-1-e},  which are also used in \cite{cycon}. As already noted in
\cite{cycon}, only minor changes are required to adapt 
the proof for a scalar operator on $M=\R^n$ to a more complicated operator on a bundle $\Eh$. 
We shall sketch the main idea for our operator $H_\hbar$ which is slightly different from the operator in \cite{cycon}. \\
Let $\chi \in \Ce^\infty_0(\R^n)$
be a cut-off function with $0 \leq \chi \leq 1$  (and such that $\sqrt{1-\chi^2}$ is smooth), let $(\Omega_j, \phi_j)$ be centered local charts based 
at $m^j \in M$ (i.e. satisfying $\phi_j(m^j)=0$) and consider the pull-back 
$\chi_j= \phi_j^*(\chi(\hbar^{-2/5} \cdot ))$ of the scaled cut-off function. For $\hbar$ sufficiently small, we have $\chi_j \in \Ce^\infty_0(M)$ (for all $1 \leq j \leq r$), and 
$\chi_0:= \sqrt{1 - \sum_{j=1}^r \chi_j^2}$ is smooth. Thus, the localization formula \eqref{Hpartition} holds. Then, on  the support of $\chi_j$, one uses 
Taylor expansion at $m^j$ of $V$ up to third order terms and of $W$ up to
 linear terms. In operator  norm, all remainder terms of the Taylor expansion
 in $\chi_j H_\hbar \chi_j$ for $1 \leq j \leq r$ are of order $ O(\hbar^{6/5})$, and so is the localization error $\hbar^2 \sum_j |d\chi_j|^2$.  
 One now fixes $n$ such that $e_n < e_{n+1}$ (for the eigenvalues of the harmonic oscillators in 
$\bigoplus H_{m^j,\hbar}$) and denotes by $g_k$, for $1 \leq k \leq n$, the corresponding normalized eigenfunctions of the appropriate harmonic 
oscillators $H_{m^j,\hbar}$, pulled back from $T_{m^j}M$ to $M$ and cut-off by multiplication with $\chi_j$. A straightforward computation then gives
 $$ \scl g_k, H_h g_k \scr_{\Eh} = \hbar e_k + O(\hbar^{\frac{6}{5}}), $$
which in view of the mini-max formula establishes the upper bound on $E_k(\hbar)$ in \eqref{0_theo7}. To establish the lower bound, one uses the mini-max 
formula again and derives a lower bound in terms of a suitable symmetric finite rank operator (constructed from the restrictions of all localized operators $H_{m^j,\hbar}$ 
to the spectral subspace of enery below $e \in (e_n,e_{n+1}$), which then implies \eqref{0_theo7}. These arguments belong to abstract spectral theory (provided the error 
bounds for localization and Taylor expansion have been established) and do not depend on the geometry encoded in $H_\hbar$ and $\Eh$.

\end{proof}

\section{Agmon-estimates}\label{sec1}

In this section we prove exponential decay of eigensections of Dirichlet realizations of $H_\hbar$, in the limit $\hbar \to 0$. 
These estimates allow to decouple the wells and are crucial for establishing good error estimates. 
Technically, the main point of our subsequent discussion is to verify that abstract properties of the metric connection $\nabla^\Eh$ are always sufficient. 
Computations in local coordinates are not required. 

\begin{prop}\label{prop1}

Let $\phi\in \Ce^0(M, \R)$ be Lipschitz, 
$E\in\R$, $\Omega\subset M$ be open and bounded and assume $v\in\Gamma_c^\infty(\Omega, \Eh)$. Then
\begin{equation}\label{prop1_eq0}
\Re \< v, e^{\frac{\phi}{\hbar}} (H_\hbar - E) e^{-\frac{\phi}{\hbar}}v \>_\Eh = \hbar^2 \bigl\| \nabla^\Eh v \bigr\|^2_{\otimes} 
+  \< v, (\hbar^2 U + \hbar W + (V - |\dd\phi|^2 - E) \Id_\Eh ) v \>_\Eh.
\end{equation}
Moreover, let $F_\pm: \Omega \rightarrow [0,\infty)$ be defined such that $F_+^2 - F_-^2 = V - |\dd\phi|^2 - E$ and $F:= F_+ + F_-$. Then 
\begin{equation}\label{prop1_eq2}
\hbar^2\bigl\| \nabla^\Eh v\bigr\|_{\otimes}^2 + \frac{1}{2} \bigl\| F_+ v\bigr\|_\Eh^2 
+ \< v, (\hbar^2 U + \hbar W) v \>_\Eh
\leq \Bigl\| \frac{1}{F}e^{\frac{\phi}{\hbar}} (H_\hbar - E) e^{-\frac{\phi}{\hbar}}v \Bigr\|_\Eh^2  + \frac{3}{2} \bigl\|F_- v\bigr\|_\Eh^2
\end{equation}
and, for some $C_1>0$,
\begin{equation}\label{prop1_eq1}
\hbar^2 \bigl\|\nabla^\Eh v\bigr\|_{\otimes }^2 + \frac {1}{4}\bigl\|F v\bigr\|_\Eh^2 \leq 
\Bigl\| \frac{1}{F}e^{\frac{\phi}{\hbar}}(H_\hbar - E) e^{-\frac{\phi}{\hbar}}v\Bigr\|_\Eh^2 + 2\bigl\|F_-v\bigr\|_\Eh^2 + 
C_1 \hbar \bigl\|v\bigr\|_\Eh^2\; .
\end{equation}
\end{prop}

\begin{proof}
In order to prove \eqref{prop1_eq0}, we first assume $\phi\in\Ce^2(M, \R)$ and use \eqref{L-Darstellung} and \eqref{schrodinger} to write
\begin{multline}
\Re \< v, e^{\frac{\phi}{\hbar}} (H_\hbar - E) e^{-\frac{\phi}{\hbar}}v \>_\Eh \\
 = \Re \Bigl\{ \hbar^2 \<\nabla^\Eh e^{\frac{\phi}{\hbar}} v, \nabla^\Eh e^{-\frac{\phi}{\hbar}} v \>_\otimes  +
\<e^{\frac{\phi}{\hbar}} v, (\hbar^2 U + \hbar W + (V - E)\Id_\Eh) e^{-\frac{\phi}{\hbar}}v \>_\Eh\, \Bigr\}.
\end{multline}
We write
\begin{equation}\label{prop1_p2}
\<\nabla^\Eh e^{\frac{\phi}{\hbar}} v, \nabla^\Eh e^{-\frac{\phi}{\hbar}} v \>_\otimes = 
\< e^{-\frac{\phi}{\hbar}}\nabla^\Eh e^{\frac{\phi}{\hbar}} v, e^{\frac{\phi}{\hbar}}\nabla^\Eh e^{-\frac{\phi}{\hbar}} v \>_\otimes
\end{equation}
and since $\nabla^\Eh (f v) = \dd f \otimes v + f \nabla^\Eh v$ for $f\in\Ce^\infty(M,\R)$ and $v\in\Gamma^\infty(M, \Eh)$, we get, using
$\dd e^{\frac{\phi}{\hbar}} =\frac{1}{\hbar} e^{\frac{\phi}{\hbar}} \dd\phi$ in the second step,
\begin{align}
\text{rhs} \eqref{prop1_p2} &=  \< e^{-\frac{\phi}{\hbar}}\Bigl(\dd e^{\frac{\phi}{\hbar}}\otimes v + e^{\frac{\phi}{\hbar}}\nabla^\Eh v\Bigr),
e^{\frac{\phi}{\hbar}}\Bigl(\dd e^{-\frac{\phi}{\hbar}}\otimes v + e^{-\frac{\phi}{\hbar}}\nabla^\Eh v\Bigr)\>_\otimes \nonumber\\
&= \< \tfrac{1}{\hbar}\dd\phi\otimes v + \nabla^\Eh v,
-\tfrac{1}{\hbar}\dd\phi\otimes v +  \nabla^\Eh v\Bigr)\>_\otimes \nonumber\\
&= 
\bigl\|\nabla^\Eh v\bigr\|_\otimes^2 - \tfrac{1}{\hbar^2}\bigl\|\dd\phi \otimes v\bigr\|_\otimes^2 + 
\< \tfrac{1}{\hbar}\dd\phi\otimes v, \nabla^\Eh v\>_\otimes -
\<  \nabla^\Eh v, \tfrac{1}{\hbar}\dd\phi\otimes v\>_\otimes\; .\label{prop1_p3}
\end{align}
Since $\Re \Bigl(\< \tfrac{1}{\hbar}\dd\phi\otimes v, \nabla^\Eh v\>_\otimes -
\<  \nabla^\Eh v, \tfrac{1}{\hbar}\dd\phi\otimes v\>_\otimes\Bigr) =0$,
equation \eqref{prop1_eq0} follows from \eqref{prop1_p3}, using 
that $\langle\!\langle v, Sv\rangle\!\rangle_\Eh$ is real for $S$ symmetric. 
The case where $\phi$ is only Lipschitz can be deduced from the above as in the scalar case (see \cite{helffer-sjostrand-1}, Prop 1.1) using convolution with a standard
mollifier and the dominated convergence theorem.\\ 

In order to prove \eqref{prop1_eq2}, we use the definition of $F_+$ and $F_-$ to write
\begin{multline}
\text{rhs} \eqref{prop1_eq0} = \hbar^2\bigl\| \nabla^\Eh v\bigr\|_\otimes^2 + 
 \< e^{\frac{1}{\hbar}\phi} u, (F_+^2 - F_-^2)\Id_\Eh e^{\frac{1}{\hbar}\phi}u \>_\Eh
+  \< v, (\hbar^2 U + \hbar W) v \>_\Eh\\
= \hbar^2\bigl\| \nabla^\Eh v\bigr\|_\otimes^2 + \bigl\| F_+ v\bigr\|_\Eh^2 - 
\bigl\|F_- v\bigr\|_\Eh^2
+  \< v, (\hbar^2 U +\hbar W) v \>_\Eh \; .   
\end{multline}
Thus we have
\begin{equation}\label{prop1_p4}
 \text{lhs}\eqref{prop1_eq2} = \text{rhs}\eqref{prop1_eq0} - \frac{1}{2}\bigl\| F_+ v\bigr\|_\Eh^2 + \bigl\|F_- v\bigr\|_\Eh^2\; .
\end{equation}
Using $ab\leq a^2 + \tfrac{b^2}{4}$ and $\frac{1}{4}(a+b)^2 \leq \frac{1}{2}(a^2 + b^2)$ we get 
\begin{align}\label{prop1_p5}
\text{lhs} \eqref{prop1_eq0} \leq \Bigl\| \frac{1}{F}e^{\frac{\phi}{\hbar}} (H_\hbar - E) e^{-\frac{\phi}{\hbar}}v \Bigr\|_\Eh \cdot \bigl\|F v\bigr\|_\Eh &\leq
\Bigl\| \frac{1}{F}e^{\frac{\phi}{\hbar}} (H_\hbar - E) e^{-\frac{\phi}{\hbar}}v \Bigr\|_\Eh^2 + \frac{1}{4} \bigl\|F v\bigr\|_\Eh^2 \text{ and}\\
\frac{1}{4}\bigl\|F v\bigr\|_\Eh^2 &\leq  \frac{1}{2} \bigl\|F_+ v\bigr\|_\Eh^2 + \frac{1}{2} \bigl\|F_- v\bigr\|_\Eh^2\label{prop1_p9}
\end{align}
Inserting \eqref{prop1_p5} and \eqref{prop1_p9} into \eqref{prop1_p4} proves \eqref{prop1_eq2}.

In order to prove \eqref{prop1_eq1}, we use \eqref{prop1_p9} to write
\begin{align}
 \text{lhs}\eqref{prop1_eq1} &\leq \hbar^2 \bigl\|\nabla^\Eh v\bigr\|_{\otimes }^2 + \frac{1}{2} \bigl\|F_+ v\bigr\|_\Eh^2 + \frac{1}{2} \bigl\|F_- v\bigr\|_\Eh^2 \nonumber\\
 &\leq \Bigl\| \frac{1}{F}e^{\frac{\phi}{\hbar}} (H_\hbar - E) e^{-\frac{\phi}{\hbar}}v \Bigr\|_\Eh^2 + 2 \bigl\|F_- v\bigr\|_\Eh^2 - \< v, (\hbar^2 U + \hbar W) v \>_\Eh \label{prop1_p8}
\end{align}
where in the second step we used \eqref{prop1_eq2}. For some $C_1>0$ (independent of $u$ and $\phi$) one has, using \eqref{1_hyp1},
\begin{equation}\label{prop1_p7}
 -\< v, (\hbar^2 U + \hbar W) v \>_\Eh \leq  C_1 \hbar \bigl\|v\bigr\|_\Eh^2\; .
\end{equation}
Inserting \eqref{prop1_p7} into \eqref{prop1_p8} proves the estimate \eqref{prop1_eq1}.
\end{proof}

We set 
\begin{equation}\label{Snull}
 S_{j,k}:= d(m^j, m^k)\quad\text{for}\quad j,k\in \mathcal{C},\, j\neq k \quad\text{and}\quad  S_0 := \min_{j,k\in \mathcal{C}, j\neq k} S_{j,k} \, .
\end{equation}

\begin{hyp}\label{Mj}
For $S_0$ given in \eqref{Snull}, there exists $S\in (0, S_0)$ such that for all $j\in\mathcal{C}$, the ball
$B_S(m^j) :=\{m\in M\,|\, d(m,m^j) < S\}$ has compact closure $\overline{B_S(m^j)}$ in $M$\footnote{Thus, in particular, any (in the topology of $M$)
closed subset of $B_S(m^j)$ is compact. Global compactness of $M$, however, is irrelevant. This is close (but not equivalent) to one of the equivalent 
statements in the Hopf-Rhinow Theorem, namely: All closed and (with respect to the $g$-distance) bounded subsets of the Riemannian manifold $B_S(m^j)$ 
are compact. Here, of course, closed is taken in the relative topology of $B_S(m^j)$ and $B_S(m^j)$ itself is bounded but not compact. 
Correspondingly, geodesics (for $g$) may reach the boundary of $B_S(m^j)$ in finite time, violating geodesic completeness of $B_S(m^j)$. 
Similarly, taking the closure of $\Gamma_c^\infty(B_S(m^j), \Eh)$ in the graph norm of $H_\hbar$ gives $H_0^2(B_S(m^j), \Eh)$ - defined by use of 
the metric connection $\nabla^\Eh$ - which is not the operator domain of the Friedrich's extension of $H_\hbar$ defined on $\Gamma_c^\infty(B_S(m^j), \Eh)$.}. 
In the following, we fix such an $S$ with the additional property that $S+\ep$ for $\ep>0$ sufficiently small still satisfies this condition. For each $j \in \mathcal{C}$, we choose a
compact manifold $M_j\subset M$ (with smooth boundary) such that $\overline{B_S(m^j)} \subset {\mathring M}_j$ and $m^k\notin M_j$ for $k\neq j$. 
Let $H_\hbar^{M_j}$ denote the operator restricted to $M_j$ with Dirichlet boundary conditions.
\end{hyp}

By standard arguments (using compact embedding theorems for Sobolev spaces), $H_\hbar^{M_j}$ has compact resolvent and thus purely discrete spectrum.

\begin{prop}\label{prop2}
For $j\in\mathcal{C}$, let $H_\hbar^{M_j}$ and $H_{m^j, \hbar}$ be given in Hypothesis \ref{Mj} and \eqref{harmoss}
respectively and 
assume that $I=[0, \hbar R_0]$, where $\hbar R_0$ is not in the spectrum of 
$H_{m^j, \hbar}$. Let $u$ be an eigenfunction of $H_\hbar^{M_j}$ with eigenvalue $E \in I$. Then there exist constants $\hbar_0, C, B>0$ such that
for all $\hbar \in (0, \hbar_0)$
\begin{equation}\label{prop2_eq0}
\hbar^2 \Bigl\|   \Bigl(\bigl( 1 + \frac{d^j}{\hbar}\bigr)^{-B} \nabla^\Eh e^{\frac{d^j}{\hbar}} u\Bigr)\Bigr\|_\otimes^2 + \hbar 
\Bigl\| \Bigl( 1 + \frac{d^j}{\hbar}\Bigr)^{-B} e^{\frac{d^j}{\hbar}}u \Bigr\|_\Eh^2 \leq C \hbar\; .
\end{equation}
\end{prop}

\begin{proof}
We fix $j\in\mathcal{C}$ and set for any $B>0$
\begin{equation}\label{prop2_p1}
\Phi (x) := \begin{cases} d^j(x) - B\hbar \ln \frac{d^j(x)}{\hbar}\, & , \quad d^j(x) > B\hbar \\ 
d^j(x) - B\hbar \ln B \, & , \quad d^j(x) \leq  B\hbar \end{cases}\, .
\end{equation}
Then for $\hbar$ sufficiently small it follows from the eikonal equation \eqref{eikonal} that
\begin{equation}\label{prop2_p2}
V(x) - |\dd\Phi(x)|^2 = V(x)  - |\dd d^j(x)|^2 = 0\, , \qquad \text{for } d^j(x) \leq B\hbar
\end{equation}
Since $\dd \Phi (x) = \bigl( 1-\frac{B\hbar}{d^j(x)}\bigr) \dd d^j(x)$ for $d^j(x) >B\hbar$, the eikonal inequality \eqref{eikonalin} yields
\begin{equation}\label{prop2_p3}
V(x) - |\dd\Phi(x)|^2 \geq V(x) \Bigl( 1 - \bigl(1-\frac{B\hbar}{d^j(x)}\bigr)\Bigr) \geq V(x) \frac{B\hbar}{d^j(x)}\geq \frac{B}{C_0}\hbar
\end{equation}
where for the last step we used that $\frac{1}{C_0} \leq \frac{V(x)}{d^j(x)} \leq C_0$ for some $C_0>0$.
In order to use \eqref{prop1_eq1}, we 
choose $B$ such that $\frac{B}{C_0}\hbar - E \geq 4\hbar (C_1 + 1)$ and
set $F=F_+ + F_-$ with
\begin{align}\label{prop2_p4}
F_+ (x) &= \id_{\{d^j\geq \hbar B\}}( V(x) - |\dd\Phi(x)|^2 - E)^{1/2} +  \id_{\{d^j < \hbar B\}} ( 4\hbar(C_1+1))^{1/2}\\
F_-(x) &= \id_{\{d^j < \hbar B\}} ( 4\hbar(C_1+1) + E)^{1/2}\nonumber
\end{align}
Then by \eqref{prop2_p2} and \eqref{prop2_p3}
\begin{equation}\label{prop2_p5}
F\geq \sqrt{4\hbar(C_1+1)} >0\,,\quad F_- = O(\sqrt{\hbar}) \quad\text{and}\quad F_+^2 - F_-^2 = V - |\dd\Phi|^2 - E,
\end{equation}
yielding
\begin{equation}\label{prop2_p7}
\frac{1}{4}\bigl\| F e^{\frac{\Phi}{\hbar}} u \bigr\|_\Eh^2 - C_1 \hbar \bigl\| e^{\frac{\Phi}{\hbar}} u \bigr\|_\Eh^2 \geq 
\hbar \bigl\| e^{\frac{\Phi}{\hbar}} u \bigr\|_\Eh^2 
\end{equation}
and since $\supp F_-=\{d^j<\hbar B\}$ for some $C>0$
\begin{equation}\label{prop2_p8}
\bigl\| F_- e^{\frac{\Phi}{\hbar}} u \bigr\|_\Eh^2 \leq C \hbar  \|u\|_\Eh^2\; .
\end{equation}
Moreover $e^{\frac{\Phi}{\hbar}}$ is of the same order of magnitude as $\bigl(1 + \frac{d^j}{\hbar}\bigr)^{-B} e^{\frac{d^j}{\hbar}}$,
thus \eqref{prop1_eq1} yields for some constants $C, C', \tilde{C}>0$
\begin{align}
\text{lhs}\eqref{prop2_eq0} &\leq 
C \Bigl(\hbar^2 \bigl\|  \nabla^\Eh e^{\frac{\Phi}{\hbar}}  u\bigr\|_\otimes^2 +  \hbar 
\bigl\| e^{\frac{\Phi}{\hbar}}u \bigr\|_\Eh^2 \Bigr)
\leq C\Bigl( \hbar^2 \bigl\| \nabla^\Eh\bigl( e^{\frac{\Phi}{\hbar}}  u\bigr)\bigr\|_\otimes^2 +  \hbar 
\bigl\| e^{\frac{\Phi}{\hbar}} u \bigr\|_\Eh^2 \Bigr) \nonumber\\
&\leq C\Bigl( \hbar^2 \bigl\| \nabla^\Eh\bigl( e^{\frac{\Phi}{\hbar}}  u\bigr)\bigr\|_\otimes^2 +   
\frac{1}{4}\bigl\| F e^{\frac{\Phi}{\hbar}} u \bigr\|_\Eh^2 - C_1 \hbar \bigl\| e^{\frac{\Phi}{\hbar}} u \bigr\|_\Eh^2\Bigr) \nonumber \\
&\leq C' \bigl\| F_- e^{\frac{\Phi}{\hbar}} u \bigr\|_\Eh^2 \leq \tilde{C} \hbar \|u\|_\Eh^2\, . \label{prop2_p9}
\end{align}

\end{proof}

\begin{cor}\label{cor1}
Under the assumptions given in Proposition \ref{prop2}, there exists $N_0\in\N$ such that 
\[
\bigl\| e^{\frac{d^j}{\hbar}} u \bigr\|_\Eh^2 +  \bigl\| \nabla^\Eh e^{\frac{d^j}{\hbar}} u \bigr\|_\otimes^2 = O\bigl(\hbar^{-N_0}\bigr)\qquad
\text{as}\quad \hbar \to 0\, .
\]
\end{cor}

\section{Interaction matrix}\label{sec3}

In this section we introduce the interaction matrix, which under appropriate spectral conditions allows to compute full
asymptotics of eigenvalue splitting on an exponentially small scale. It is our main technical tool.

We start with notations and recall some standard facts.
By $\Lambda^p_\C (M) = \Lambda^p(M)\otimes \C$ we denote the complexified exterior bundle; its smooth sections are the complex differential $p$-forms  in $\Omega_\C^p(M)$.

We extend the hermitian form $\gamma$ to a sesquilinear fibrewise pairing
$$ \gamma:  \Eh \times (T^*_\C M \otimes \Eh) \to \Lambda^1_\C (M), \qquad 
\gamma[u,\alpha \otimes v] := \gamma[u,v] \alpha,$$ 
where $u,v$ are in $\Eh_m$ and $\alpha \in T^*_{\C,m} M$ which in the standard way extends to the full tensor product by linearity. Similarly, we define an extension in the first factor, giving 
$$ \gamma: ((T^*_\C M  \otimes  \Eh)    \times  \Eh  \to \Lambda^1_\C (M), \qquad 
\gamma[\alpha \otimes u, v] := \gamma[u,v] \overline{\alpha} ,$$ 

 We feel free to (often) suppress the subscript $m$. 
In particular, for $u,v\in \Gamma^\infty (M, \Eh)$, 
we then have 
$ \gamma [\nabla^\Eh u, v]\in \Omega^1_\C(M)$, satisfying 
\begin{equation}\label{form-gamma}  
\gamma [\nabla^\Eh u, v] (X) = \gamma [\nabla^\Eh_X u, v]\, , \quad X\in  \Gamma(TM)\,, 
\end{equation}
and similarly for $ \gamma [ u, \nabla^\Eh v]$.
Since $\nabla^\Eh$ is a metric connection,
\[
\dd \gamma [u,v] (X) = X(\gamma [u,v]) = \gamma [\nabla^\Eh_X u, v] + \gamma [u,\nabla^\Eh_X v]   , \quad X\in  \Gamma(TM)\,. 
\]
Combined with \eqref{form-gamma} this yields
\begin{equation}\label{form1}
\dd \gamma [u,v] =  \gamma [\nabla^\Eh u, v] + \gamma [u,\nabla^\Eh v],
\end{equation}
which is just an equivalent (and shorter) expression for $\nabla^\Eh$ being metric.
As in the real case, the Hodge star operator $*:\Omega^p_\C (M) \rightarrow \Omega^{n-p}_\C (M)$ associates to any $\mu\in \Omega^p_\C (M)$ 
the $(n-p)$-form 
$*\mu$ given by 
\begin{equation}  \label{hodgestar}
 *\mu_m(v_{p+1}, \ldots v_n) = \mu_m (v_1, \ldots v_p)\, , \qquad m\in M,
\end{equation}
where $v_1, \ldots v_n$ are oriented orthonormal vectors in $T_mM$. 
In particular, if $i: \Sigma \to M$ denotes an embedded regular hypersurface in $M$, oriented by its outer normal field $N$ we have
\begin{equation}  \label{oneform}
i^*(* \omega)= \omega(N) d\sigma,\qquad \omega\in \Omega^1_\C(M),
\end{equation}
where $d \sigma$ is the induced volume form on $\Sigma$.
Furthermore,  for each $m\in M$ an inner 
product on 
$\Lambda^p_\C (M)$ is defined fibrewise by
\begin{equation}\label{innerproductp}  
\langle \mu,\nu \rangle \, \dvol (m) := \overline{\mu} \wedge *\nu
\end{equation}
where $\dvol$ denotes the volume form associated to $g$. Thus we can define an $L^2$-inner product on the compactly supported sections of $\Lambda^p_\C(\Omega)$ for $\Omega \subset M$ by 
\begin{equation}\label{innerproductpM} 
\scl \mu, \nu \scr_{p,\Omega} := \int_\Omega \overline{\mu} \wedge *\nu = \int_\Omega \langle \mu,\nu \rangle\, \dvol \, .
\end{equation}
In the case $\Omega = M$ we shall simply drop the subscript.

\begin{Lem}\label{lemma1}
Let $\Omega$ be an $n$-dimensional submanifold of $M$ with smooth boundary $\partial \Omega$, then for any $\beta\in \Lambda^{k-1}_\C(M)$ and 
$\gamma\in \Lambda^k_\C(M) $ with compact support
\begin{equation} \label{lemma1_0}
\int_{\partial \Omega} \overline{\beta}\wedge *\gamma = \scl \dd\beta, \gamma \scr_{k,\Omega} - \scl \beta, \delta \gamma \scr_{k-1, \Omega} 
\end{equation}
where 
\begin{equation}\label{codiff}
\delta := (-1)^{n(k+1) +1} *\dd*: \Lambda^k_\C(M) \rightarrow \Lambda^{k-1}_\C(M) 
\end{equation}
denotes the codifferential operator.
\end{Lem}

We remark that, if the boundary $\partial \Omega$ is empty, equation  \eqref{lemma1_0} shows in particular that 
$\delta$ actually coincides with the adjoint $ \dd^*$, i.e.
\begin{equation}\label{dconjugiert}
\scl \dd\beta, \gamma \scr_{k,\Omega} = \scl \beta, \delta \gamma \scr_{k-1, \Omega} \, , \qquad \beta\in \Omega^{k-1}_\C(M), \gamma\in \Omega^k_\C(M).
\end{equation}

\begin{proof}
By Stokes Theorem, we have 
\[ \int_{\partial \Omega} \overline{\beta}\wedge *\gamma = \int_\Omega \dd(\overline{\beta}\wedge *\gamma) =
\int_\Omega \Bigl[(\dd\overline{\beta}) \wedge *\gamma + (-1)^{k-1} \overline{\beta}\wedge (\dd*\gamma)\Bigr] \]
Using $**\alpha = (-1)^{k(n-k)} \alpha$ for any $\alpha\in\Lambda^k_\C(M)$ and $\dd*\gamma \in \Lambda^{n-k+1}_\C(M)$, the right hand side is equal to

\[ \scl \dd\beta, \gamma\scr_{k,\Omega} + (-1)^{k-1 + (n-k+1)(k-1) - n(k+1) -1}\int_\Omega \overline{\beta}\wedge *\delta\gamma = 
\scl \dd\beta, \gamma \scr_{k,\Omega} - \scl \beta, \delta\gamma \scr_{k-1, \Omega}\, ,\]
since the exponent of $(-1)$ is equal to $2(k-n-1) - k(k-1) - 1$ and thus impair.
\end{proof}

The following hypothesis ensures that there is no spectrum exponentially close to the boundary of the spectral interval we shall consider later on.

\begin{hyp}\label{hyp1}
For $M_j,\,  j\in\mathcal{C}$ as given in Hypothesis \ref{Mj}, let $I_\hbar = [\alpha (\hbar),\beta (\hbar)]$ be an interval, such that
$\alpha (\hbar),\beta (\hbar)$ are $O(\hbar)$ for $\hbar\to 0$. Furthermore we assume that there
exists a function $a(\hbar)>0$ with the property $|\log a(\hbar)| =
o\left(\frac{1}{\hbar}\right),\, \hbar\to 0$, such that none of the
operators $H_\hbar,H_\hbar^{M_1},\ldots H_\hbar^{M_r}$ given in Hypotheses \ref{setup1} and \ref{Mj} has spectrum in
$[\alpha(\hbar)-2a(\hbar),\alpha(\hbar)[$ or $]\beta(\hbar),\beta(\hbar)+2a(\hbar)]$.
\end{hyp}

Given a spectral interval $I_\hbar$ as above, let
\begin{eqnarray}\label{specHepusw}
\spec (H_\hbar) \cap I_\hbar = \{ \lambda_1,\ldots , \lambda_N\} \,
,&\quad&
u_1,\ldots ,u_N\in L^2(M, \Eh)\\
\F := \Span \{u_1,\ldots u_N\} \nonumber\\
\spec \left(H_\hbar^{M_j}\right) \cap I_\hbar = \{ \mu_{j,1},\ldots,
\mu_{j,n_j}\} \, ,
&\quad& v_{j,1},\ldots,v_{j,n_j}\in L^2 (M_j, \Eh) ,\, j\in {\mathcal C} \nonumber\\
\end{eqnarray}
denote the eigenvalues of $H_\hbar$ and of the Dirichlet operators
$H_\hbar^{M_j}$ inside of the spectral interval  $I_\hbar$ and corresponding orthonormal systems
of eigenfunctions. We write 
\begin{equation}\label{valpha}
 v_\alpha\quad\text{with}\quad \alpha =(\alpha_1, \alpha_2)\in \mathcal{J}:=\{(j,\ell)\,|\,j\in\mathcal{C},\, 1\leq \ell \leq n_j\}
\quad \text{and}\quad j(\alpha):= \alpha_1\; .
\end{equation}

Let $\chi_j\in \Ce_0^\infty(M_j,\R)$ be such that $\chi_j=1$ in an open neighborhood of $\overline{B_S(m^j)}$. We set
\begin{equation}\label{defpsi}
\psi_{j,\ell} := \chi_j v_{j,\ell}\, , \quad 
\E_j := \Span \{ \psi_{j,1},\ldots, \psi_{j,n_j} \} \quad\text{and}\quad  \E := \bigoplus \E_j \, .
\end{equation}

For closed subspaces $\E$
and $\F$ of any Hilbert space $\mathscr{H}$, we denote by $\Pi_{\E}$ and $\Pi_{\F}$ the orthogonal
projections on $\E$ and $\F$ respectively. Then we define the
nonsymmetric distance $\vec{\dist}(\E,\F)$ between $\E$ and $\F$
by
\[ \vec{\dist}(\E,\F) := \| \Pi_{\E} - \Pi_{\F}\Pi_{\E} \| \; .\]

The following theorem is analog to \cite{helffer-sjostrand-1}, Theorem 2.4, Lemma 2.8 and is crucial for the construction of the interaction matrix.

\begin{theo} \label{theo1}
Under the assumptions given in Hypotheses \ref{setup1}, \ref{Mj} and \ref{hyp1} and with the notation given above, 
there exists $S_2 \in (S, S_0)$ such that for $\hbar$ sufficiently small and for all $s<S_2$ and for all $\alpha\in \mathcal{J}$
\ben
\item $H_\hbar \psi_\alpha = \mu_\alpha \psi_\alpha + O\Bigl(e^{-\frac{S_2}{\hbar}}\Bigr)$ in $L^2(M_{j(\alpha)}, \Eh)$.
\item $\vec{\dist}(\E,\F) = \vec{\dist}(\F,\E) = O\Bigl(e^{-\frac{s}{\hbar}}\Bigr)$.
\item there exists a bijection $b: \sigma (H_\hbar)\cap I_\hbar \rightarrow \bigcup_{j=1}^r\sigma(H_\hbar^{M_j}) \cap I_\hbar$ such that 
$|b(\lambda) - \lambda| = O\Bigl(e^{-\frac{s}{\hbar}}\Bigr)$.
\item $\|\Pi_0 - \Pi_\E\| = O\Bigl(e^{-\frac{s}{\hbar}}\Bigr)$ where $\Pi_0$ denotes the projection onto $\E$ along $\F^\perp$. 
\een
\end{theo}

\begin{proof} 
In order to show (a), we write 
\begin{equation}\label{theo1_1} 
H_\hbar \psi_\alpha = \mu_\alpha \psi_\alpha + \big[H_\hbar, \chi_{j(\alpha)}\big] v_\alpha 
\end{equation} 
and observe that 
\begin{equation}  \label{theo1_2}
\bigl[H_\hbar, \chi_{j(\alpha)}\bigr]= \hbar^2  \bigl[ \big(\nabla^\Eh\bigr)^*\nabla^\Eh, \chi_{j(\alpha)}\bigr] \, .
\end{equation}
We claim that, for any  $\chi\in \Ce_0^\infty(M,\R)$,   one has as operator on $\Gamma^\infty (M, \Eh)$ 
\begin{equation}  \label{new2}
\bigl[\big(\nabla^\Eh\bigr)^* \nabla^\Eh, \chi \bigr] = (\Delta \chi) - 2 
\nabla^\Eh_{\grad \chi}
\end{equation}
where $\Delta=\dd^* \dd= \delta d$ 
(see \eqref{codiff}). 

In fact, dropping momentarily  for reasons of brevity the superscript $\Eh$ in $\nabla^\Eh$, one readily computes
\begin{equation} \label{new3}
\bigl[\nabla^* \nabla, \chi \bigr]= T - T^*,
\end{equation}
where
\begin{equation} \label{new4}
T= [\nabla^*, \chi] \nabla \quad\text{and}\quad T^*= - \nabla^* (\dd \chi \otimes\quad).
\end{equation}
We need a few identities on linear algebra in the fibres of $\Eh$ and $T^*_\C M \otimes \Eh$ which we include for the convenience of the reader. In view of the usual musical isomorphism $\langle \dd \chi, \o\rangle_1= \o(\grad \chi)$,
which is standard at least for {\em real} $\chi$, one gets the first equality in
\begin{equation}  \label{musical}
\gamma[u, \nabla_{\grad \chi} v] =\langle \dd \chi, \gamma[u,\nabla v ]\rangle_1 =
\langle \dd \chi \otimes u, \nabla v \rangle_{\otimes},
\end{equation}
for $u,v \in \Gamma^\infty (M, \Eh)$.     
For the second equality, it suffices to consider the special case $\nabla v= \alpha \otimes w$ for $\alpha \in T^*_{\C,m} M$ and $w \in \Eh_m$ and then use linearity for the general case. But in this case
$$ \mbox{rhs } \eqref{musical} = \langle \dd \chi, \gamma[u,w]\alpha \rangle_1 =  
\langle \dd \chi, \gamma[u,\alpha \otimes w] \rangle_1, $$
establishing \eqref{musical}.
We now find
\begin{equation} \label{new5}
T= -  \nabla_{\grad \chi},
\end{equation} 
since for smooth sections $u,v$ of $\Eh$ (at least one compactly supported) one has,
momentarily dropping the subscript $\Eh$,
\begin{equation}
\langle\!\langle Tu,v \rangle\!\rangle  = \langle\!\langle 
\nabla u,  - \dd \chi \otimes v \rangle\!\rangle_\otimes 
 = - \langle\!\langle  \nabla_{\grad  \chi} u,v \rangle\!\rangle,
\end{equation}
using \eqref{musical} (complex conjugated and integrated over $M$) in the last step.
Similarly, 
using both the first and second equality in \eqref{musical}  combined with $ \nabla$ being metric, one finds
\begin{multline}  \label{new7}
\scl -T^* u,v \scr = \scl \dd \chi \otimes u, \nabla v \scr_\otimes = \scl \dd \chi, \gamma[u,\nabla v] \scr_1
 = \scl \dd \chi,\dd (\gamma[u,v])  -\gamma[\nabla u,v] \scr_1 \\
 = \scl (\dd^* \dd \chi) u,v \scr - \scl \nabla_{\grad \chi}  u,v \scr .
\end{multline}
This actually proves
\begin{equation} \label{new8}
-T^* =   \Delta -\nabla_{\grad \chi},
\end{equation}
and combining equation \eqref{new8} with \eqref{new5} gives \eqref{new2}.
 Using the commutator formula \eqref{new2} we can now estimate the commutator term on the rhs of \eqref{theo1_1}. 
The assumption on the cut-off function $\chi_j$ gives that $d^{j} \geq S_1 $ on the support of $\dd \chi_{j}$ for some $S_1\in (S, S_0)$. Thus one finds
\begin{equation}  \label{theo1_13}
\big\| \big[H_\hbar, \chi_{j(\alpha)}\big] v_\alpha \big\|_\Eh 
\leq 
C e^{-\frac{S_1}{\hbar}} \Bigl( \bigl\|e^{\frac{d^{j(\alpha)}}{\hbar}} v_\alpha\bigr\|_\Eh + \bigl\|e^{\frac{d^{j(\alpha)}}{\hbar}} \nabla v_\alpha\bigr\|_\otimes\Bigr)\, 
\end{equation}  
for some $C>0$,  since the estimate on the summand involving $\Delta \chi_j$, where $j=j(\alpha)$,  is trivial and the summand involving 
$\nabla_{\grad \chi_j}$ satisfies, using \eqref{musical},
\begin{align}
\Big| \scl u,    \nabla_{\grad \chi_{j}} v_{\alpha} \scr_\Eh \Big| &
\leq \Big| \scl e^{-\frac{d^{j}}{\hbar}} \dd \chi_{j} \otimes u,  e^{\frac{d^{j}}{\hbar}} \nabla v_{\alpha} \scr_\otimes \Big| 
\leq e^{-\frac{S_1}{\hbar}} \bigl\|\dd \chi_{j} \otimes u\bigr\|_\otimes\bigl\|e^{\frac{d^{j}}{\hbar}} \nabla v_\alpha\bigr\|_\otimes \nonumber\\
& \leq C e^{-\frac{S_1}{\hbar}}\bigl\|u\bigr\|_\Eh \bigl\|e^{\frac{d^{j}}{\hbar}} \nabla v_\alpha\bigr\|_\otimes \label{theo1_9}
\end{align}
for some $ C >0$. Taking the supremum over $\|u\|=1$ gives \eqref{theo1_13}.
Since
\[ \bigl\| e^{\frac{d^{j}}{\hbar}} \nabla^\Eh v_\alpha \bigr\|_\otimes \leq \bigl\| \nabla^\Eh\bigl(e^{\frac{d^{j}}{\hbar}} v_\alpha\bigr) \bigr\|_\otimes + 
\bigl\| e^{\frac{d^{j}}{\hbar}} \tfrac{1}{\hbar}\dd d^{j}\otimes  v_\alpha \bigr\|_\otimes,  \]
equation \eqref{theo1_13} together with Corollary \ref{cor1}  shows (a) for any $S_2< S_1$.

In order to prove (b), we recall that at each well $m^j,\, j\in \mathcal{C},$ the Dirichlet eigenfunctions $v_{j,\ell}$ are orthonormal and thus
by Corollary \ref{cor1} and since $d^{j(\alpha)} + d^{j(\beta)} \geq S_{j(\alpha),j(\beta)}$
\begin{equation}\label{theo1_11}
\scl v_\alpha, v_\beta\scr_\Eh = \delta_{\alpha\beta} + \bigl(1-\delta_{j(\alpha) j(\beta)}\bigr) O\bigl(\hbar^{-N_0}e^{-\frac{S_{j(\alpha),j(\beta)}}{\hbar}}\bigr)\, .
\end{equation}
for some $N_0\in \N$. Since moreover $d^{j} \geq S_1$ on the support of $1- \chi_{j}$ for some $S_1\in (S, S_0)$ it follows from \eqref{theo1_11}, using again Corollary \ref{cor1}, that for some $N_0\in \N$ and 
for all $\alpha,\beta\in\mathcal{J}$
\begin{equation}\label{theo2_11} 
\langle\!\langle \psi_\alpha, \psi_\beta\rangle\!\rangle_\Eh = \delta_{\alpha \beta} + \delta_{j(\alpha) j(\beta)} O\bigl(\hbar^{-N_0}e^{-\frac{2 S_1}{\hbar}}\bigr) 
+ \bigl(1-\delta_{j(\alpha) j(\beta)}\bigr) O\bigl(\hbar^{-N_0}e^{-\frac{S_{j(\alpha),j(\beta)}}{\hbar}}\bigr) \, . 
\end{equation}
Then the proofs of (b) and (c) proceed exactly along the lines of \cite{helffer-sjostrand-1}, Theorem 2.4.

(d) can be seen as in \cite{helffer-sjostrand-1}, Lemma 2.8.
\end{proof}

In the following theorem, we introduce the notion of interaction matrix, refining the analysis of the error term above.

\begin{theo}\label{theo2}
In the setting of Theorem \ref{theo1}, for all $s < S_2$ and $\hbar$ sufficiently small, the matrix of $\Pi_0 H_\hbar|_\E$ in the basis $\psi_\alpha, \, \alpha\in \mathcal{J},$ is given by
\begin{equation}\label{malphabeta} 
\Bigl( m_{\alpha\beta}\Bigr)_{\alpha,\beta\in\mathcal{J}} +
O\left(e^{-\frac{2 s}{\hbar}}\right)\, , \qquad m_{\alpha\beta} = \delta_{\alpha\beta}\mu_\alpha + w_{\alpha\beta} 
\end{equation}
where the interaction matrix is given by 
\begin{align}\label{intact}
w_{\alpha\beta} &= \< \psi_\alpha, [H_\hbar, \chi_{j(\beta)}] v_\beta\>_\Eh \\
&= \hbar^2 \Bigl( \< \chi_{j(\alpha) }\nabla^\Eh v_\alpha, \dd\chi_{j(\beta)} \otimes v_\beta\>_\otimes - 
\< \chi_{j(\alpha)} \dd\chi_{j(\beta)} \otimes v_\alpha, \nabla^\Eh v_\beta\>_\otimes \Bigr) + O\left(e^{-\frac{2 s}{\hbar}}\right)
\, .\label{intact2}
\end{align}
In particular, we have for some $N_0\in\N$
\begin{equation}\label{theo2_0} 
w_{\alpha\beta}=\delta_{j(\alpha) j(\beta)} O\bigl(e^{-\frac{2 s}{\hbar}}\bigr) + \bigl(1-\delta_{j(\alpha) j(\beta)}\bigr) O\bigl(\hbar^{-N_0}e^{-\frac{S_{j(\alpha),j(\beta)}}{\hbar}}\bigr)\, .
\end{equation}
\end{theo}

\begin{proof}
 Writing
\begin{equation}\label{theo2_1}
\Pi_0 H_\hbar \psi_\beta = \mu_\beta \psi_\beta + B \psi_\beta\quad \text{where}\quad 
B\psi_\beta:= \Pi_0 \bigl[ H_\hbar, \chi_{j(\beta)}\bigr] v_\beta
\end{equation}
the statement on the matrix representation means that we have to determine $w_{\alpha\beta}$ such that 
\begin{equation}\label{theo2_2}
B \psi_\beta = \sum_{\alpha\in \mathcal{J}} w_{\alpha\beta }\psi_\alpha + O\bigl(e^{-\frac{2s}{\hbar}}\bigr)\, .
\end{equation}
For any ordered basis of a finite dimensional Hilbert space, we use the notation
\begin{equation}\label{theo2_10}
\vec{x}:= (x_1, \ldots, x_n) \quad\text{and}\quad \Ge_{\vec{x}}:= \vec{x}^*\vec{x} = \Big( \langle x_j, x_k\rangle \Bigr)_{1\leq j,k\leq n}\, .
\end{equation}
Then, setting $\vec{\psi}:= (\psi_{1,1}, \psi_{1,2} \ldots, \psi_{r, n_r-1}, \psi_{r,n_r})$ for the basis of the Hilbert space 
$\E=\bigoplus_{j\in{\mathcal C}} \E_j$, estimate \eqref{theo2_11} yields
\begin{equation}\label{theo2_6}
\Ge_{\vec{\psi}}=: \id + T\quad\text{with}\quad T=: \bigl(t_{\alpha\beta}\bigr)_{\alpha,\beta\in \mathcal{J}} = O\bigl(e^{-\frac{s}{\hbar}}\bigr)
\end{equation}
for any $s< S_1$. We define 
\begin{equation}\label{theo2_5}
\tau := \vec{\psi}\vec{\psi}^* =  \sum_{\alpha\in\mathcal{J}} \langle\!\langle\psi_\alpha, \,\cdot\,\rangle\!\rangle_\Eh \psi_\alpha\;.
\end{equation}
Since $\Ge_{\vec{\psi}}$ is self-adjoint and positive, an orthonormal system of $\E$ is given by 
\begin{equation}\label{theo2_9}
\vec{\phi}:= \vec{\psi}\Ge_{\vec{\psi}}^{-\frac{1}{2}}\, ,
\end{equation}
and \eqref{theo2_6} yields
\begin{equation}\label{theo2_9a}
\Pi_\E = \vec{\phi}\vec{\phi}^* = \vec{\psi}\Ge_\psi^{-1}\vec{\psi}^* = \tau + O\bigl(e^{-\frac{\sigma}{\hbar}}\bigr)\, .
\end{equation}
Thus, combining equation \eqref{theo2_9a}  and Theorem \ref{theo1}(d)
we get
$\|\Pi_0 - \tau\| = O\bigl(e^{-\frac{s}{\hbar}}\bigr)$ for any $s<S_2$.
Together with Theorem \ref{theo1}(a), this yields
\begin{equation}\label{theo2_7}
\Pi_0 \bigl[ H_\hbar, \chi_{j(\beta)}\bigr] v_\beta = 
 \sum_{\alpha\in\mathcal{J}} \<\psi_\alpha,\bigl[ H_\hbar, \chi_{j(\beta)}\bigr] v_\beta\>_\Eh 
\psi_\alpha + O\bigl(e^{-\frac{2s}{\hbar}}\bigr)
\end{equation}
and therefore by \eqref{theo2_1} and \eqref{theo2_2} we get \begin{equation}\label{theo2_12}
\Pi_0 H_\hbar \vec{\psi} = \vec{\psi} M  + O\bigl(e^{-\frac{2s}{\hbar}}\bigr) \quad\text{for} \quad M=\bigl(\delta_{\alpha\beta}\mu_\alpha + w_{\alpha\beta}\bigr)_{\alpha,\beta\in \mathcal{J}}
\end{equation}
for $w_{\alpha\beta}$ given in \eqref{intact}. 
This proves \eqref{malphabeta} and \eqref{intact}. To see \eqref{intact2}, we write
\begin{align}
w_{\alpha\beta} &= \hbar^2 \<\psi_\alpha,\bigl[  \big(\nabla^\Eh\bigr)^*\nabla^\Eh, \chi_{j(\beta)}\bigr] v_\beta\>_\Eh\nonumber\\
&= \hbar^2 \Big\{ \<\nabla^\Eh \bigl(\chi_{j(\alpha)} v_\alpha\bigr),\nabla^\Eh \bigl( \chi_{j(\beta)} v_\beta\bigr)\>_{\otimes} - 
\<\nabla^\Eh \bigl(\chi_{j(\beta)}\chi_{j(\alpha)} v_\alpha\bigr),\nabla^\Eh v_\beta\>_{\otimes} \Bigr\}\label{theo2_8}
\end{align}
Using product rule, some of the terms cancel and we get
\begin{align*}
\text{rhs}\eqref{theo2_8} &= \hbar^2 \Bigl\{ \<\dd\chi_{j(\alpha)}\otimes v_\alpha,\dd\chi_{j(\beta)}\otimes v_\beta\>_{\otimes} +
\<\chi_{j(\alpha)}\nabla^\Eh v_\alpha, \dd\chi_{j(\beta)} \otimes v_\beta\>_{\otimes}\\
&\quad - 
\<\chi_{j(\alpha)}\dd\chi_{j(\beta)}\otimes v_\alpha\bigr),\nabla^\Eh v_\beta\bigr)\>_{\otimes}\Bigr\}\, .
\end{align*}
Using again that $d^j\geq S_1$ on the support of $\dd\chi_j$, it follows from Corollary \ref{cor1} that the first term on the right hand side is $O \bigl(e^{-\frac{2s}{\hbar}}\bigr)$. This proves \eqref{intact2}.
Equation \eqref{theo2_0}  follows from \eqref{intact2}, using \eqref{theo1_11} together with Corollary \ref{cor1} and the fact $d^j>S_2$ on the support of $\dd\chi_j$.
\end{proof}

Since $H_\hbar$ is self-adjoint, it should have a symmetric matrix representation. Moreover, we want to give a 
matrix representation for $H_\hbar|\F$. 

\begin{theo}\label{theo3}
In the setting of Theorem \ref{theo1}, let $\vec{\phi}$ denote the orthonormalization of $\vec{\psi}$ in $\E$ as given in \eqref{theo2_9}, choose 
$\vec{f}:= \Pi_\F \vec{\phi}$ as basis in $\F$ 
and denote by $\vec{g}:= \vec{f} \Ge_{\vec{f}}^{-\frac{1}{2}}$ its orthonormalization. Then for all $s < S_2$ and $\hbar$ sufficiently small the matrix of 
$H_\hbar|_\F$ with respect to $\vec{g}$ is 
\begin{equation}\label{theo3_0}
\tilde{M} + O\Bigl(e^{-\frac{2s}{\hbar}}\Bigr)\quad\text{with}\quad \tilde{M} = \big(\tilde{m}_{\alpha\beta}\bigr) = 
\Bigl( \mu_\alpha \delta_{\alpha\beta} \Bigr) + \frac{1}{2} \Bigl( w_{\alpha\beta} + \overline{w}_{\beta\alpha} \Bigr) 
\end{equation}
for $w_{\alpha\beta}$ given in \eqref{intact}.
\end{theo}

\begin{proof}
First we compute
\begin{align}
w_{\alpha\beta} - \overline{w}_{\beta\alpha} &= \scl \psi_\alpha, [H_\hbar, \chi_{j(\beta)}] v_\beta\scr_\Eh - 
\overline{\scl \psi_\beta , [H_\hbar, \chi_{j(\alpha)}] v_\alpha\scr}_\Eh\nonumber\\
&= \scl \psi_\alpha, H_\hbar \psi_\beta\scr_\Eh - \scl \psi_\alpha, \chi_{j(\beta)} \mu_\beta v_\beta\scr_\Eh -
\overline{\scl \psi_\beta , H_\hbar \psi_\alpha\scr}_\Eh + \overline{\scl \psi_\beta , \chi_{j(\alpha)} \mu_\alpha v_\alpha\scr}_\Eh\nonumber\\
&= \bigl(\mu_\alpha - \mu_\beta\bigr) \scl\psi_\alpha, \psi_\beta\scr_\Eh = \bigl(\mu_\alpha - \mu_\beta\bigr)  t_{\alpha\beta}\label{theo3_1}
\end{align}
where in the last step we used \eqref{theo2_6}. By \eqref{theo2_12} and \eqref{theo2_10} we can write 
$\vec{\psi}^*\Pi_0H_\hbar \vec{\psi} \equiv \Ge_{\vec{\psi}} M$ where here and in the following $\equiv$ is equality 
modulo $ O\Bigl(e^{-\frac{2s}{\hbar}}\Bigr)$. 
Since $\vec{\phi}$ is orthonormal, by \eqref{theo2_9} and \eqref{theo2_6} the matrix of $\Pi_0 H_\hbar|_\E$ with respect to $\vec{\phi}$ is  given by
\begin{align}
\hat{M}:= \vec{\phi}^*\Pi_0H_\hbar \vec{\phi} &\equiv \Ge_{\vec{\psi}}^{\frac{1}{2}} M \Ge_{\vec{\psi}}^{-\frac{1}{2}} = 
(\id + T)^{\frac{1}{2}} M (\id + T)^{-\frac{1}{2}}\nonumber\\
&\equiv (\id + \tfrac{1}{2}T) M (\id - \tfrac{1}{2}T) \equiv M + \tfrac{1}{2}[T,M] \equiv M + \tfrac{1}{2} [T, \diag(\mu_\alpha)] \label{theo3_2}
\end{align}
where we used Taylor expansion and that both $T$ and $(w_{\alpha\beta})$ are of order $ O\Bigl(e^{-\frac{s}{\hbar}}\Bigr)$. 
By \eqref{theo3_1} we can write
\begin{align}
\text{rhs}\eqref{theo3_2} &= \bigl(\mu_\alpha \delta_{\alpha\beta}\bigr) + \bigl(w_{\alpha\beta}\bigr) + 
\tfrac{1}{2}\bigl(t_{\alpha\beta} (\mu_\beta - \mu_\alpha)\bigr)\nonumber\\
&= \bigl(\mu_\alpha \delta_{\alpha\beta}\bigr) + \bigl(w_{\alpha\beta}\bigr) -
\tfrac{1}{2}\bigl(w_{\alpha\beta} - \overline{w}_{\beta\alpha}\bigr) = \bigl(\mu_\alpha \delta_{\alpha\beta}\bigr) + 
\tfrac{1}{2}\bigl(w_{\alpha\beta} + \overline{w}_{\beta\alpha}\bigr) \, .\label{theo3_7}
\end{align}
Since $\Pi_0$ is the projection on $\E$ along $\F^\perp$, we have $\ker \Pi_0 = \ker \Pi_\F$ and $\Pi_\F \Pi_0 = \Pi_\F$ and the eigenspaces $\E$ and $\F$ are in bijection  
via $\Pi_0|_\F$ and $\Pi_\F|_\E$. Moreover $\F$ and $\F^\perp$ are invariant under the action of $H_\hbar$ and therefore $\Pi_\F H_\hbar = H_\hbar \Pi_\F$,
thus 
\begin{equation}\label{theo3_3}
H_\hbar \vec{f} =  H_\hbar \Pi_\F \vec{\phi} = \Pi_\F \Pi_0 H_\hbar \vec{\phi} = \Pi_\F \vec{\phi}\hat{M} = \vec{f}\hat{M}
\end{equation}
where we used that by \eqref{theo2_9a} $\Pi_0 = \Pi_\E \Pi_0 = \vec{\phi}\vec{\phi}^* \Pi_0$ and the definition of $\hat{M}$.
Writing $\phi_\alpha = f_\alpha + h_\alpha$ for $f_\alpha\in \F$ and $h_\alpha\in\F^\perp$, we get by Theorem \ref{theo1}
\begin{equation}\label{theo3_4}
\bigl\|h_\alpha\bigr\|_\Eh = \bigl\|\phi_\alpha - f_\alpha\bigr\|_\Eh = \bigl\|\bigl(\Pi_\E - \Pi_\F\Pi_\E\bigr)\phi_\alpha\bigr\|_\Eh \leq \vec{\dist}(\E, \F) 
= O\bigl(e^{-\frac{s}{\hbar}}\bigr)
\end{equation}
and therefore
\begin{equation}\label{theo3_5}
 \id = \bigl( \langle\!\langle \phi_\alpha, \phi_\beta\rangle\!\rangle\bigr)_{\alpha,\beta} =\bigl( \langle\!\langle f_\alpha, f_\beta\rangle\!\rangle + 
\langle\!\langle h_\alpha, h_\beta\rangle\!\rangle\bigr)_{\alpha,\beta} = 
\Ge_{\vec{f}} + O\bigl(e^{-\frac{2s}{\hbar}}\bigr)\; .
\end{equation}
Thus, analog to \eqref{theo3_2}, using \eqref{theo3_3}, \eqref{theo3_5} and that $\vec{g}$ is orthonormal, the matrix of $H_\hbar|_\F$ in the basis 
$\vec{g}$ is given by
\begin{equation}\label{theo3_6}
\tilde{M} = \vec{g}^*H_\hbar \vec{g}= \Ge_{\vec{f}}^{-\frac{1}{2}} \vec{f}^* H_\hbar \vec{f}\Ge_{\vec{f}}^{-\frac{1}{2}} \equiv
\Ge_{\vec{f}}^{\frac{1}{2}}\hat{M}\Ge_{\vec{f}}^{-\frac{1}{2}} \equiv \hat{M}
\end{equation}
where $\equiv$ means equality modulo $O\Bigl(e^{-\frac{2s}{\hbar}}\Bigr)$. Combining \eqref{theo3_6} with \eqref{theo3_2} and \eqref{theo3_7} 
proves the theorem.
\end{proof}

As in \cite{helffer-sjostrand-1}, Thm.2.12, it follows that

\begin{cor}\label{cor2}
For $\hbar$ sufficiently small, there is a bijection
\[ b: \sigma (H_\hbar|_\F) \longrightarrow \sigma(\tilde{M}) \quad\text{such that}\quad |b(\lambda) - \lambda| = O\Bigl(e^{-\frac{2s}{\hbar}}\Bigr)\, . \]
\end{cor}

\section{Interaction matrix in special cases}\label{2wells}

In this section we give an explicit formula for the interaction matrix element $w_{\alpha\beta}$ in the case that the two wells
$m^{j(\alpha)}, m^{j(\beta)}$ are near and the Dirichlet operators have very close eigenvalues inside the chosen spectral 
interval $I_\hbar$.
We start with some properties of the one-form $\gamma [\nabla^\Eh u, v]$ introduced in \eqref{form-gamma}.

\begin{Lem}\label{lemma2}
For $u, v\in\Gamma^\infty (M, \Eh)$ and $\delta$ the codifferential operator defined in \eqref{codiff}
\begin{equation}\label{lemma2_7}
\delta \gamma [ u, \nabla^\Eh v] = \gamma [ u,  (\nabla^\Eh)^*\nabla^\Eh v] - \langle\nabla^\Eh u, \nabla^\Eh v\rangle_\otimes \in \Ce^\infty (M, \C)\, .
\end{equation}

\end{Lem}

\begin{proof}
First we recall that $\delta=\dd^*$  (since $M$ has no boundary). For reasons of brevity we drop the superscript (and later the subscript)  $\Eh$.
Thus Lemma \ref{lemma1} yields for any $\phi\in\Ce_0^\infty (M, \C)$ 
\begin{equation}\label{lemma2_1}
\scl \phi, \delta \gamma [\nabla u, v] \scr_0 = \scl \dd\phi, \gamma [\nabla u, v] \scr_{1} = \scl \dd \phi \otimes u, \nabla v \scr _\otimes,
\end{equation}
using \eqref{musical}.
Moreover, 
for such $\phi$ we have
\begin{equation} \label{lemma2_4}
\scl \phi u, \nabla^* \nabla v \scr = \scl \dd \phi \otimes u, \nabla v \scr_\otimes + 
\scl \phi \nabla u, \nabla v \scr_\otimes.
\end{equation}

Combining \eqref{lemma2_1} and  \eqref{lemma2_4} one gets
\begin{equation}
\scl \phi, \mbox{lhs } \eqref{lemma2_7} \scr_0 = \scl \phi, \mbox{rhs } \eqref{lemma2_7} \scr_0,
\end{equation}
which finishes the proof since $\phi$ was arbitrary.
\end{proof}

Similarly we get, using an equation analogous to \eqref{musical}, with $u$ and $\nabla v$ interchanged as arguments of $\gamma$,
\begin{equation}\label{lemma2_0}
\delta \gamma [\nabla^\Eh u, v] = \gamma [(\nabla^\Eh)^*\nabla^\Eh u, v] - \langle\nabla^\Eh u, \nabla^\Eh v\rangle_\otimes \in \Ce^\infty (M, \C)
\end{equation}

We now give assumptions leading to a more explicit form for the interaction matrix.

\begin{hyp}\label{hypalphabeta}
 Under the assumptions given in Hypotheses \ref{setup1}, \ref{Mj} and \ref{hyp1} and with the notation given at the beginning of Section \ref{sec3}
 we assume that $\alpha, \beta\in\mathcal{J}$ are pairs such that for some constant $0<a < 2S - S_0$ 
\begin{equation}\label{prop4_00}
S_{j(\alpha),j(\beta)}< S_0 + a \qquad\text{and}\qquad |\mu_\alpha - \mu_\beta |= \expord{a} \; .
\end{equation}
Setting $j=j(\alpha), k=j(\beta)$ to shorten the notation, we define the closed ``ellipse'' 
\begin{equation}\label{prop4_01}
 G_{j,k} := \{ m\in M\,|\, d^j(m)+ d^k(m) \leq S_0 + a \}
\end{equation}
We remark that $ G_{j,k} $ is contained in the union $B_S(m^j) \cup B_S(m^k)$ which is compact by assumption. Thus,
in particular, $ G_{j,k} $ (and $\Sigma_{j,k}$ to be defined below) are compact in $\stackrel{\circ}{M}_{j}\cup\stackrel{\circ}{M}_{k}$.
We choose $\Omega_{j,k}\subset M$ open with smooth boundary, such that 
\begin{equation} \label{prop4_02}
 m^j\in\Omega_{jk}\, , \quad m^k\notin \overline{\Omega}_{j,k}\, , \quad G_{j,k}\cap \overline{\Omega}_{j,k} \subset \stackrel{\circ}{M}_{j}\, , 
\quad G_{j,k} \cap \Omega_{j,k}^c \subset \stackrel{\circ}{M}_{k}
\end{equation}
and set $\Sigma_{j,k} := \partial \Omega_{j,k} \cap G_{j,k}$.
\end{hyp}

The following proposition gives an explicit formula for the interaction term by means of a surface integral.

\begin{prop}\label{prop4}
Under the assumptions on the pairs $\alpha, \beta\in \mathcal{J}$ given in Hypothesis \ref{hypalphabeta} the elements $w_{\alpha\beta}$ of the 
interaction matrix, modulo $O\Bigl(\hbar^{-N_0}e^{-\frac{1}{\hbar}(S_0 + a)}\Bigr)$ for some $N_0\in\N$,  are given by
 \begin{align}\label{prop4_0}
 \frac{1}{\hbar^2} w_{\alpha\beta} &\equiv \int_{\Sigma_{j,k}} *\omega_{\alpha\beta} 
\\
  &=  \int_{\Sigma_{j,k}} \Bigl( \gamma_m[\nabla^\Eh_N v_\alpha, v_\beta] - 
\gamma_m[v_\alpha, \nabla^\Eh_N v_\beta] \Bigr)\, d\sigma (m)  \label{prop4_0b}
   \end{align}
where  
 $\omega_{\alpha\beta}\in \Lambda^1_\C(M)$ is defined by
\begin{equation}\label{lemma2a_0}
 \omega_{\alpha\beta}:= \gamma [ \nabla^\Eh v_\alpha,  v_\beta ] - \gamma [ v_\alpha, \nabla^\Eh v_\beta]\; .
 \end{equation}
 and $N$ is the outward unit normal on $\partial\Omega_{j,k}$, i.e. the unit normal on $\Sigma_{j,k}$ pointing from $m^j$ to $m^k$.
\end{prop}

\begin{proof}
 We fix a pair $\alpha,\beta\in\mathcal{J}$ satisfying \eqref{prop4_00} and write $G=G_{j,k}$, $\Omega = \Omega_{j,k}$ and $\Sigma=\Sigma_{j,k}. $
 Let $\chi_G\in \Ce_0^\infty (M)$ be a cut-off function such that $\chi_G =1$ on $G$ and with support close to $G$, then 
$\supp\chi_G \cap \overline{\Omega} \subset \stackrel{\circ}{M}_{j}$ and $\supp\chi_G \cap \Omega^c \subset \stackrel{\circ}{M}_{k}$.

We choose the cut-off functions $\chi_j$ and $\chi_k$ in the definition of $\psi_\alpha$ and $\psi_\beta$ (see \eqref{defpsi}) such that 
$\chi_j=1$ on $\supp\chi_G \cap \overline{\Omega}$ and $\chi_k=1$ on $\supp\chi_G \cap \Omega^c$.
 
 Then the definition of $G$ together with Corollary \ref{cor1} (the exponential decay of the Dirichlet eigenfunctions) allow 
 modulo $O\Bigl(\hbar^{-N_0}e^{-\frac{1}{\hbar}(S_0 + a)}\Bigr)$ for some $N_0\in\N$ (which we denote by $\equiv$) to insert the additional 
 cut-off function $\chi_G$ into the formula \eqref{intact2} for $w_{\alpha\beta}$, thus (using also that $\chi_j=1$ in $\Omega\cap \supp \chi_G$ by the 
assumptions above, using \eqref{musical} and dropping the superscript $\Eh$)
 \begin{align}
 \frac{1}{\hbar^2} w_{\alpha\beta} &\equiv 
\scl \chi_G \nabla v_\alpha, \dd \chi_k \otimes v_\beta\scr_\otimes
- \scl \chi_G \dd \chi_k \otimes v_\alpha, \nabla v_\beta \scr_\otimes\\
& =\scl \dd \chi_k, \o \scr_{1,\Omega}, \qquad \omega := \chi_G \, \omega_{\alpha\beta}
  \label{prop4_1}
 \end{align}
Lemma \ref{lemma1} together with \eqref{form1} and $\chi_k=1$ on $\Sigma$ leads to 
\begin{equation}\label{prop4_3}
  \frac{1}{\hbar^2} w_{\alpha\beta} = \int_{\Sigma} *\omega + (\chi_k, \delta \omega)_{0, \Omega} 
\end{equation}
Since $\chi_G=1$ on $\Sigma$, the first term on the right hand side of \eqref{prop4_3} is equal to the right hand side of \eqref{prop4_0}.
We now claim that
 \begin{equation}\label{lemma2a_1}
  (\chi_k, \delta \omega)_{0, \Omega}  = O\Bigl(\hbar^{-N_0}e^{-\frac{1}{\hbar}(S_0 + a)}\Bigr)\, ,
   \end{equation}
   proving \eqref{prop4_0}. To see \eqref{lemma2a_1} we first note that for some constant $C>0$
   \[
  \bigl|  (\chi_k, \delta \omega)_{0, \Omega}\bigr| \leq C \bigl\|\delta \omega\bigr\|_\Eh
  \]
 We use \eqref{lemma2_0} and \eqref{lemma2_7} to write
    \begin{equation}\label{prop4_4}
\bigl\|\delta \omega\bigr\|_\Eh = \bigl\| \delta\gamma [\nabla^\Eh v_\alpha, \chi_G v_\beta] -  \delta\gamma [\chi_G v_\alpha,  \nabla^\Eh v_\beta] \bigr\|_\Eh 
\leq A_1 + A_2 
 \end{equation}
where 
\begin{align*}
A_1 &=  \bigl\|  \gamma [(\nabla^\Eh)^*\nabla^\Eh v_\alpha, \chi_G v_\beta] - \gamma [ \chi_G v_\alpha, (\nabla^\Eh)^*\nabla^\Eh v_\beta]\bigr\|_\Eh \\
A_2 &= \bigl\| \langle \nabla^\Eh (\chi_G v_\alpha) , \nabla^\Eh v_\beta \rangle_{\otimes} - \langle \nabla^\Eh v_\alpha , \nabla^\Eh (\chi_G v_\beta) \rangle_{\otimes} \bigr\|_\Eh
\end{align*}
By product rule
\begin{equation}\label{prop4_6}
A_2 \leq \bigl| \scl \dd \chi_G\otimes v_\alpha , \nabla^\Eh v_\beta\scr_\otimes \bigr| + \bigl| \scl  \nabla^\Eh v_\alpha ,\dd \chi_G\otimes v_\beta\scr_\otimes \bigr| \equiv 0
\end{equation}
where the last estimate follows from the fact that $d^j + d^k > S_0 + a$ on the support of $\dd\chi_G$ together with the exponential decay properties of $v_\alpha$ and $v_\beta$ (Corollary \ref{cor1}).
To analyze $A_1$, we remark that by Hypothesis \ref{setup1} the endomorphism fields $U$ and $W$ are symmetric on $\Eh$ and $\hbar^2 U + \hbar W + V\Id_\Eh$ commutes with $\chi_G \Id_\Eh$, thus
\begin{equation}\label{prop4_7}
A_1 = \bigl\|  \gamma [ H_\hbar v_\alpha, \chi_G v_\beta] - \gamma [ \chi_G v_\alpha, H_\hbar v_\beta]\|_\Eh = 
|\mu_\alpha - \mu_\beta| \bigl| \scl  v_\alpha, \chi_G v_\beta \scr_\Eh\bigr| = O\bigl(e^{-\frac{a}{\hbar}}\bigr)O\bigl(\hbar^{N_0} e^{-\frac{S_{j,k}}{\hbar}}\bigr) \equiv 0\, .
\end{equation}
The last two estimates follow from assumption \eqref{prop4_00} together with Corollary \ref{cor1}. Inserting \eqref{prop4_6} and \eqref{prop4_7} into \eqref{prop4_4} proves \eqref{lemma2a_1} and thus  \eqref{prop4_0}.

Applying \eqref{oneform} to the hypersurface $\Sigma$ proves \eqref{prop4_0b}.
\end{proof}

\begin{rem}\label{rem1}
\ben
\item 
If $S_{jk}>S_0+a$ and $|\mu_\alpha - \mu_\beta |= \expord{a}$, then it follows at once from \eqref{theo2_0} that 
$w_{\alpha\beta}=O\bigl(\hbar^{-N_0}e^{-\frac{S_0 + a}{\hbar}}\bigr)$. Thus the formula \eqref{prop4_0} is 
relevant only if the Agmon distance $S_{jk}$ between the wells and the difference of the 
Dirichlet eigenvalues in the assumptions of Proposition \ref{prop4} are related by \eqref{prop4_00}.\\
If  $a$ is large, then $S_{jk}$ is nearly $2S_0$, but $|\mu_\alpha-\mu_\beta|$ must be very small.
If on the contrary $a$ is small, then $S_{jk}$ must be near to $S_0$, but $|\mu_\alpha-\mu_\beta|$ is comparatively large (though still exponentially small of 
course). 
\item 
It is possible to treat the limiting case
\begin{equation}\label{prop4_00a}
 d(m^{j(\alpha)}, m^{j(\beta)})=S_0\qquad\text{and}\qquad |\mu_\alpha-\mu_\beta|=O\bigl(\hbar^\infty\bigr)
\end{equation}
along the lines of the above proof,
choosing $a$ in the construction of $\Sigma_{jk}$ arbitrarily small,
yielding
\begin{equation}\label{prop4_0a}
 \frac{1}{\hbar^2} w_{\alpha\beta} =  \int_{\Sigma_{j(\alpha),j(\beta)}} \Bigl( \gamma_m[\nabla^\Eh_N v_\alpha, v_\beta] - 
\gamma_m[v_\alpha, \nabla^\Eh_N v_\beta] \Bigr)\, d\sigma (m) + 
O\Bigl(\hbar^{\infty}e^{-\frac{1}{\hbar}S_0}\Bigr)
   \end{equation}
where $N$ is the outward unit normal on $\partial\Omega_{j(\alpha),j(\beta)}$.
\een
\end{rem}

\section{Asymptotic expansion}

Using the quasimodes for the Dirichlet operators constructed in \cite{ludewig}, we will give asymptotic expansions for the interaction term $w_{\alpha\beta}$ in the
case considered in Section \ref{2wells}.

We start with some additional hypotheses:

\begin{hyp}\label{hypomega}
Let $X_{\tilde{h}_0}$ denote the Hamiltonian vector field on $T^*M$ with respect to $\tilde{h}_{0}$ defined in \eqref{tildehnull}, 
$F_{t}$ denote its flow and for $j\in\mathcal{C}$ set
\begin{equation}\label{Lambdaplus}
\Lambda_{\pm}^j := \bigl\{ (m,\xi)\in T^*M\, |\, F_{t}(m,\xi) \rightarrow (m^j,0)\quad \text{for}
\quad t \rightarrow \mp \infty \bigr\}  \; .
\end{equation}
Let $M_j$ satisfy Hypothesis \ref{Mj}. We assume that there is $\O^j\subset\subset M_j$, open and containing $m^j$, such that the following
holds.
\ben
\item For $\tau : T^*M \rightarrow M$ denoting the bundle projection $\tau (m, \xi) = m$, we have 
\[ \Lambda_+(\Omega^j):=\tau^{-1}(\Omega^j) \cap \Lambda_+^j = \bigl\{ (m, \dd d^j(m)) \in T^*M\, |\, m\in \Omega^j\bigr\} \; . \]
\item $\tau\bigl(F_t(m,\xi)\bigr) \in \Omega^j$ for all $(m,\xi)\in \tau^{-1}(\Omega^j) \cap \Lambda_+^j$ and all $t\leq0$. 
\een
\end{hyp}

By \cite{kleinro}, Theorem 1.5, the base integral curves of $X_{\tilde{h}_0}$ on
$M\setminus\{m^1,\ldots m^r\}$ with energy $0$ are geodesics with respect to
$d$ and vice versa.
Thus the above hypothesis implies in particular that there is a unique minimal geodesic
between any point in $\Omega^j$ and $m^j$.

Clearly, $\Lambda_+ (\Omega^j)$ is a Lagrange manifold (by (a)) and since the flow $F_t$ preserves $\tilde{h}_0$, we have
$\Lambda_+(\Omega^j)\subset \tilde{h}_0^{-1}(0)$ by \eqref{Lambdaplus}. Thus the eikonal equation 
$\tilde{h}_0(m, \dd d^j(m)) =0$ holds for all $m\in\Omega^j$.
Since in our setting the (in $\Omega^j$) unique solution $d^j(\cdot)$ of the eikonal equation is defined by following 
the flow of the Hamiltonian field and projecting to the base $\Omega^j$, it follows that in fact $d^j\in \Ce^\infty(\O^j)$.

Geometrically speaking,
Hypothesis \ref{hypomega} (a) means that $\Lambda_{\pm}\bigl(\Omega^j\bigr)$ projects diffeomorphically to $\Omega^j$.\\

The projection of $X_{\tilde{h}_0}$, evaluated on $\Lambda_+ (\Omega^j)$, onto the configuration space $\Omega^j$ is given by
$\partial_\xi \tilde{h}_0 (m, \xi = \dd d^j(m)) = 2 \grad d^j(m)$. Thus the pair $(d^j, \Omega^j)$ is, for each $j\in\mathcal{C}$, an admissible pair in the sense of
\cite{ludewig}, Def. 2.6, i.e. $d^j$ is the unique non-negative solution of the eikonal equation $|\dd d^j(m)|^2 = V(m)$ for $m\in \O^j$ and
$\Phi_t(\O^j)\subset \O^j$ for all $t\leq 0$, where $\Phi_t$ denotes the flow of the vector field $2\grad d^j$.
In particular, $\Omega^j$ is star-shaped with respect to the vector field $2\grad d^j$.

By straightforward calculations (compare \cite{ludewig}) we have on $\Omega^j$

\begin{equation} \label{LocalFormHPhi}
  H_{d^j,\hbar} := e^{d^j/\hbar} H_\hbar e^{-d^j/\hbar} = \hbar^2 L + \hbar \bigl( \nabla^\Eh_{2\grad d^j} + W + \Delta d^j \bigr)
\end{equation}
where $\nabla^\Eh$ is the unique metric connection determined by $L$ given in \eqref{L-Darstellung} and $\Delta$ denotes the Laplace-Beltrami operator 
acting on functions.

The next theorem is a version of the results given in \cite{ludewig}, Theorem 2.7 and Corollary 2.10, adapted to the case of more than one potential well. 

\begin{theo} \label{theo4}
Let $H_\hbar$ be as described in Hypothesis \ref{setup1}. For $j\in\mathcal{C}$ let $\O^j$ and $M_j$ satisfy Hypothesis \ref{hypomega} and fix $K$ compact in $\Omega^j$. 
Furthermore, we assume that $\hbar E_j$ denotes an eigenvalue of multiplicity $\ell_j$ of the local harmonic oscillator 
$H_{m^j, \hbar}$ at $m^j$ as given in \eqref{harmoss}. Then, for $\hbar_0$ sufficiently small and for $ \alpha=(j, k)\,,\; k =1,\ldots,\ell_j\, , \;
\ell\in\frac{\mathbb{Z}}{2}\,,\; \ell \geq -N_j$ for some $N_j\in \frac{\N}{2}$, there are functions
$a_{\alpha} \in C^\infty((0, \hbar_0), \Gamma_c^\infty(M, \Eh))$ and sections $a_{\alpha,\ell} \in \Gamma_c^\infty(M, \Eh)$, compactly supported in 
$\Omega^j$, such that for all
$N\in \frac{\mathbb{Z}}{2}$ there are
$C_N <\infty$ satisfying
\begin{equation}\label{arep}
\Bigl| a_\alpha(m; \hbar) - \sum_{\natop{\ell\in\hZ}{-N_j\leq \ell\leq N}}  \hbar^\ell
 a_{\alpha,\ell}(m)\Bigr| \leq C_N \hbar^{N+\frac{1}{2}}\, , \quad (m\in M)
\end{equation}
and
real functions $E_\alpha(\hbar)$ with asymptotic expansion
\begin{equation}\label{Erep}
E_\alpha(\hbar) \sim \hbar \Bigl(E_j + \sum_{s\in\frac{\N^*}{2}} \hbar^s E_{\alpha,s}\Bigr) \qquad \text{as}\quad \hbar\to 0
\end{equation}
such that 
the following holds:
\ben
\item 
the sections
\begin{equation}\label{hatvjk}
 \hat{v}_\alpha(\hbar):= \hbar^{-\frac{n}{4}} e^{-\frac{d^{j}}{\hbar}}  a_\alpha(\hbar)
\end{equation}
are approximate eigensections for $H_\hbar$ with respect to the approximate eigenvalues 
$E_\alpha(\hbar)$ given in \eqref{Erep}, i.e.
\begin{equation}
H_\hbar  \hat{v}_\alpha(\hbar) -  E_\alpha(\hbar) \hat{v}_\alpha(\hbar) = \begin{cases} o(\hbar^\infty)\quad \text{uniformly on} \;M\\
o(\hbar^\infty e^{-d^{j(\alpha)}/\hbar})\quad \text{uniformly on}\; K
\end{cases}
\end{equation}
\item for $\alpha=(j,k), \beta=(i,\ell)$ as above, the approximate eigensections given in \eqref{hatvjk} are 
almost orthonormal in the sense that 
\begin{equation}\label{orthokont}
\scl\hat{v}_\alpha(\hbar), \hat{v}_\beta(\hbar)\scr_\Eh = \delta_{\alpha\beta} + \delta_{ji}O(\hbar^\infty) + 
O\bigl( \hbar^{-(N_i+N_j+\frac{n}{2})} e^{-\frac{S_{ji}}{\hbar}}\bigr)\; .
\end{equation}
\een
\end{theo}

\begin{rem}
 With the notation given in \eqref{LocalEigenvalues}, the lowest order in $\hbar$ in the expansion of $a_\alpha$ is given by $N_{j(\alpha)} = \max_\gamma |\gamma|/2$ where $\gamma$ runs over all 
multi-indices such that $e^j_{\gamma,\ell} = E_j$ for some $\ell=1, \dots \mathrm{rk}\,\Eh$.
\end{rem}

\begin{prop}\label{prop5}
Let $I_\hbar = [0, \hbar R_0 ]$ for some $R_0>0$. Then for $j\in\mathcal{C}$ and $\hbar$ sufficiently small there is a bijection $b: \sigma\bigl(H_\hbar^{M_j}\bigr)\cap I_\hbar\rightarrow 
\sigma\bigl(H_{m^j, \hbar}\bigr)\cap I_\hbar$ and a
constant $C_0>0$ such that $|b(\lambda) - \lambda| \leq C_0 \hbar^{3/2}$. 
\end{prop}

\begin{proof}
Combine Theorem \ref{theo7} on the harmonic approximation and Theorem \ref{theo1} on the existence of a bijection between the spectrum of $H_\hbar$ 
and the union of the spectra of $H_\hbar^{M_j}$, both restricted to a spectral interval $I$ (giving the existence of $b$ and a rough bound $O(\hbar^{6/5})$) with
Theorem \ref{theo4} on the existence of asymptotic expansions (which improves the rough bound to $O(\hbar^{3/2})$). 
\end{proof}

Now we will prove that the difference between the quasimodes of Theorem \ref{theo4} and the Dirichlet eigensection is exponentially small.

\begin{theo}\label{theo5}
Let $H_\hbar$ be given in Hypothesis \ref{setup1} and for any $j\in\mathcal{C}$, let $\O^j, M_j$ satisfy
Hypothesis \ref{hypomega}. 
Furthermore, we assume that $\hbar E_j$ denotes an eigenvalue of $H_{m^j, \hbar}$ defined in
\eqref{harmoss} with multiplicity $\ell_j$ and we set 
$I_\hbar \bigl(E_j\bigr) = \bigl[\hbar E_j - C_0 \hbar^{\frac{3}{2}}, \hbar E_j + C_0 \hbar^{\frac{3}{2}}\bigr]$
for some $C_0>0$.
For $\alpha = (j,k),\, 1\leq k \leq \ell_j,$ let $v_\alpha$ denote orthonormal eigensections of the Dirichlet operator $H_\hbar^{M_j}$ 
with eigenvalue belonging  to the spectral interval $I_\hbar\bigl(E_j\bigr)$. Let $K$ be any compact set in $\Omega^j$ and let $\hat{v}_\alpha$ (resp. $E_\alpha$) be the quasimodes (resp. the approximate eigenvalues) associated to $\hbar E_j$, 
as defined in Theorem \ref{theo4}, and denote by 
$\mathcal{J}_j$  the set (of cardinality $\ell_j$) of all such $\alpha$.

Then there is a  unitary  $\ell_j \times \ell_j$ matrix $C^j(\hbar) = \bigl(c^j_{\alpha,\beta}(\hbar)\bigr)_{\alpha, \beta \in I_j}$ - possesing a full asymptotic expansion in half-integer powers of $\hbar$ - such that for $\hbar$ sufficiently small and $\alpha \in \mathcal{J}_j$
\begin{equation}\label{theo5_0}
v_\alpha = \tilde{v}_\alpha + O(\hbar^\infty)\quad \text{where}\quad \tilde{v}_\alpha := \sum_{\beta \in I_j} c^j_{\alpha,\beta}(\hbar) \hat{v}_{\beta}\; .
\end{equation}
Moreover for any $N\in\N$ 
\begin{equation}\label{theo5_1} 
\left\| \id_K e^{\frac{d^j}{\hbar}} \left(v_\alpha-\tilde{v}_\alpha\right)
\right\|^2_\Eh + \left\| \id_K \nabla^\Eh \bigl(e^{\frac{d^j}{\hbar}} (v_\alpha - \tilde{v}_\alpha)\bigr)\right\|^2_{\otimes} = O\left(\hbar^N\right) \; , 
\qquad (\hbar\to 0) \, .
\end{equation} 

\end{theo}

We remark that we can choose $c^j_{\alpha,\beta}(\hbar)=0$ if $E_{\beta}(\hbar)$ is not 
asymptotically equal to $\mu_{\alpha}(\hbar)$ (the Dirichlet eigenvalue associated to $v_{\alpha}$), thus $\bigl(c^j_{}\bigr)_{\alpha,\beta \in I_j}$ can be chosen to be the identity matrix if all 
$E_{\alpha}(\hbar), \, \alpha \in \mathcal{J}_j,$ have different expansions. \\
Note, furthermore,  that in view of standard elliptic estimates the basic estimate
\eqref{theo5_1} establishes similar bounds on all higher derivatives:  Second order derivatives of $e^{\frac{d^j}{\hbar}} \left(
v_\alpha-\tilde{v}_\alpha\right)$ can be bounded from the elliptic  equation and \eqref{theo5_1}, and mathematical induction then implies bounds on all derivatives. 
In particular, the Sobolev embedding theorem on the compact subset $\overline{\O^j}$ of $M$ (where all the standard definitions of Sobolev spaces actually coincide) give the following result: 
If $H$ denotes an oriented hypersurface in $\O^j$ and $d\sigma$ the induced surface measure, then \eqref{theo5_1}  implies
\begin{equation}  \label{surface}
\int_{H \cap K} e^{\frac{2d^j}{\hbar}} \gamma[v_\alpha-\tilde{v}_\alpha,v_\alpha-\tilde{v}_\alpha] d\sigma +
 \int_{H \cap K} |\nabla^\Eh \bigl(e^{\frac{d^j}{\hbar}} (v_\alpha - \tilde{v}_\alpha)\bigr)|_{\otimes}^2 d\sigma  = O\left(\hbar^\infty\right)  
\qquad (\hbar\to 0) \, ,
\end{equation}
where $| \cdot |_{\otimes}$ denotes the norm in the fibres of $T^*M \otimes \Eh$ induced from $\langle \cdot, \cdot \rangle_{\otimes}$.

We also remark that similar considerations apply to the Agmon estimates in Section 3, yielding in particular 

\begin{equation}
\int_{H \cap K} e^{2 d_j/\hbar} \gamma[v_\alpha,v_\alpha] d\sigma +
\int_{H \cap K} |\nabla^\Eh (e^{d_j/\hbar} v_\alpha)|_{\otimes} d\sigma = O\left(\hbar^{-N_0}\right),  
\end{equation}
for some $N_0 \in \N$, using Corollary \ref{cor1}.

\begin{proof}
Here one may follow the arguments in \cite{helffer-sjostrand-1}, Theorem 5.8 (the scalar case).
Denoting by $\tilde{\E}_j$ and $\hat{\E}_j$ the space spanned by $v_{j,k}, \, 1\leq k \leq \ell_j$ and $\hat{v}_{j,k}, \, 1\leq k \leq \ell_j$ respectively, it follows from Theorem \ref{theo4} and Proposition \ref{prop5}
(as in \cite{helffer-sjostrand-1}, using Proposition 2.5 of that paper) that 
\begin{equation}\label{theo5_2}
\vec{\dist}(\tilde{\E}_j, \hat{\E}_j) = \vec{\dist}(\hat{\E}_j, \tilde{\E}_j) = O\bigl(\hbar^\infty)\quad \text{and}\quad \mu_{j,k}= \hbar E_{j,k}(\hbar)+ O(\hbar^\infty)
\end{equation}
where $\mu_{j,k}$ denotes the eigenvalues of $H_\hbar^{M_j}$ associated to $v_{j,k}$. 
This proves \eqref{theo5_0}.\\

From Corollary \ref{cor1} and \eqref{hatvjk} it is clear that the left hand side of \eqref{theo5_1} is of order $O(\hbar^{-N_0})$ for some $N_0\in\N$.
In order to simplify the notation, we fix $\alpha=(j,k)\in\mathcal{J}$ and set  
\begin{equation}\label{theo5_3}
r:= \bigl( H_\hbar - \mu_\alpha\bigr) w\, , \qquad w:= v_\alpha- \tilde{v}_\alpha
\end{equation}
Then Theorem \ref{theo4} shows for any compact $\tilde{K} \subset \mathring{\Omega}^j$ (fixed in advance as amplified in Theorem \ref{theo5})
\begin{equation} \label{theo5_4}
 \Bigl\|\id_{\tilde{K}} e^{\frac{d^j}{\hbar}} r \Bigr\|_{\Eh} = O\bigl(\hbar^\infty\bigr)\; .
\end{equation}
Furthermore, by \eqref{theo5_0} we have 
\begin{equation}\label{theo5_5}
 \| \id_{\tilde{K}} w \|_{\Eh} = O(\hbar^\infty)\; .
\end{equation}
Let $\chi\in\Ce^\infty_0(\Omega^j)$ be a cut-off function, which is equal to one in a neighborhood of the union
$\hat{K}$ of all minimal geodesics from points in $K$ to $m^j$. For $\Phi$ defined in \eqref{prop2_p1} we set for $N\in \N$ and $\ep>0$
\begin{equation}\label{theo5_6}
\Phi_N(m) := \min \bigl\{  \Phi(m) + N\hbar \ln \tfrac{1}{\hbar},  \Psi(m) \bigr\} \quad\text{where}\; \Psi(m):=\inf_{n\in\supp \dd \chi} \Phi(n) + (1-\ep) d(m,n)\, .
\end{equation}
Then a compactness argument (see \cite{helffer-sjostrand-1}, Lemma 5.7) shows that if $U$ is a small neighborhood of $\hat{K}$ and $\ep$ is sufficiently 
small, then for each $N$ there
exists $\hbar_N>0$ such that $\Phi_N(m)=\Phi(m) + N\hbar \ln \tfrac{1}{\hbar}$ for all $\hbar<\hbar_N$ and $m\in U$. 
On the other hand, for any 
$m, m'\in \Omega^j$ we have 
$$|\Psi(m)-\Psi(m')| \leq (1-\ep) d(m, m')=(1-\ep)\sqrt{V(m)}d_g(x,x')(1+o(1))\qquad (x'\to x)$$ 
where $d_g$ denotes the distance with
respect to the Riemannian metric. Thus 
$$|\dd\Psi|^2 \leq (1-\ep) V\quad\text{and}\quad V(m) - |\dd\Psi(m)|^2\geq (2\ep-\ep^2) V(m) \geq \ep C$$ 
for some constant $C>0$ and
for $m$ in a region bounded away from $m^j$. For $\ep$ sufficiently small (independent of $N$), we therefore get for some $C_0>0$
\begin{equation}\label{theo5_7}
V(m) - |\dd\Phi_N(m)|^2 \begin{cases} = 0\quad\text{if}\quad d^j(m) < B\hbar \\ \geq \frac{B\hbar}{C_0}\quad\text{if}\quad d^j(m)\geq B\hbar\; .
\end{cases}
\end{equation}
We choose $B$ such that 
$\frac{B}{C_0} - \frac{\mu_\alpha}{\hbar} \geq 1$ and 
define $F_+$ and $F_-$ as in \eqref{prop2_p4}, replacing $\dd\Phi$ by $\dd \Phi_N$ and $E$ by $\mu_{j,k}$. 
We remark that $e^{d^j/\hbar} = O\bigl(\hbar^{-N_0}e^{\Phi/\hbar}\bigr)$ for some $N_0\in\N$ and $\Phi_N = \Phi + N\hbar\ln\hbar^{-1}$ in $K$.  
Using also that $\id_K[\nabla^\Eh, \chi] = 0$, we have for some 
$C>0$
\begin{multline}
\text{lhs}\eqref{theo5_1} \leq C\hbar^{2(N-N_0)} \Bigl(\bigl\| \nabla^\Eh\bigl( e^{\frac{\Phi_N}{\hbar}} \chi w\bigr)\bigr\|^2_\otimes + 
 \bigl\| e^{\frac{\Phi_N}{\hbar}} \chi w\bigr\|^2_\Eh\Bigr) \\
\leq C\hbar^{2(N-N_0)} \Bigl(\bigl\| \nabla^\Eh\bigl( e^{\frac{\Phi_N}{\hbar}} \chi w\bigr)\bigr\|^2_\otimes +  
\tfrac{1}{4\hbar} \bigl\| F e^{\frac{\Phi_N}{\hbar}} \chi w\bigr\|^2_\Eh - C_1 \bigl\| e^{\frac{\Phi_N}{\hbar}} \chi w\bigr\|^2_\Eh\Bigr)\label{theo5_8}
\end{multline}
where in the second step we used (the analog of ) \eqref{prop2_p7}. From \eqref{prop1_eq1} it follows that
\begin{align}
\text{rhs}\eqref{theo5_8} &\leq C\hbar^{2(N-N_0-1)} \Bigl( \bigl\| \frac{1}{F}e^{\frac{\Phi_N}{\hbar}} (H_\hbar - \mu_{j,k})  \chi w\bigr\|^2_\Eh + 
2  \bigl\| F_- e^{\frac{\Phi_N}{\hbar}} \chi w\bigr\|^2_\Eh \Bigr) \nonumber\\
&\leq C\hbar^{-2N_0-3} \bigl\| e^{\frac{\Phi}{\hbar}} \chi (H_\hbar - \mu_{j,k})  w\bigr\|^2_\Eh + 
C \hbar^ {2(N-N_0)-3} \bigl\| e^{\frac{\Phi}{\hbar}} [H,\chi] w\bigr\|^2_\Eh  \nonumber\\
&\quad +
2 C\hbar^{-2N_0-1} \bigl\| \id_{\{d^j<\hbar B\}} e^{\frac{\Phi}{\hbar}} \chi w\bigr\|^2_\Eh \label{theo5_9}
\end{align}
where for the last step we used that $e^{\frac{\Phi_N}{\hbar}} \leq e^{\frac{\Phi}{\hbar}}  \hbar^ {-N}$ for the first and third term 
and the fact that $\Phi_N\leq \Phi$ on the support of $\dd \chi$ (by the definition 
of $\Psi$) together with 
Corollary \ref{cor1}
\eqref{new2} for the second term.

Choosing $\tilde{K} = \supp \chi$, the last term on the right hand side of \eqref{theo5_9} is $O(\hbar^\infty)$ by \eqref{theo5_5}, the first term is
$O(\hbar^\infty)$ by \eqref{theo5_4}.  Since $\|e^{\Phi/\hbar}[H,\chi]w\|^2 = O(\hbar^{-N_1})$ for some $N_1\in\N$ by the definition of $\Phi$ and Corollary \ref{cor1}, this proves 
\eqref{theo5_1}.

\end{proof}

We shall now combine the approximate Dirichlet eigensections with the formula \eqref{prop4_0} (or \eqref{prop4_0a}). 
Under more special conditions, we shall refine the construction at the beginning of Section \ref{sec3}. We  start by giving appropriate  assumptions for the index-set $\mathcal{J}$ of the relevant 
set of Dirichlet eigenfunctions and derive an associated spectral interval. 

\begin{hyp}\label{hyp4}
Let $H_\hbar$ be given in Hypothesis \ref{setup1} and for any $j\in\mathcal{C}$, let $M_j, H_\hbar^{M_j}$ and $S$ satisfy
Hypothesis \ref{Mj}. 
\ben
\item[1)]
Let $\hbar E_0$ be in the spectrum of the direct sum  of the localized harmonic oscillators $H_{m^j,\hbar},\, j\in\mathcal{C},$ given in \eqref{harmoss}.
Let $\mathcal{J}$ be a maximal set of pairs $\alpha = (j,k)$ 
such that for
$\alpha\in\mathcal{J}$ all asymptotic eigenvalues $E_\alpha (\hbar)$ of $H_\hbar$ given in
Theorem \ref{theo4} with leading order $\hbar E_0$ are equal. Let $\mu_\alpha$ be the associated eigenvalues and $v_\alpha$ be the eigensections of the Dirichlet operators
$H_\hbar^{M_{j(\alpha)}}$. \\
\item[2)] Assume that \eqref{prop4_00} or \eqref{prop4_00a} holds for all 
$\alpha,\beta\in\mathcal{J}$. 
We choose $\alpha\in\mathcal{J}$, $N_0\in\N$ and $C>0$, such that the interval 
$I := [\mu_\alpha - C \hbar^{N_0}, \mu_\alpha + C\hbar^{N_0}]$ exactly includes the eigenvalues $\mu_\beta$ of Dirichlet operators with $\beta\in\mathcal{J}$
(this is possible for $N_0$ sufficiently small because of the maximality of $\mathcal{J}$).
\item[3)] For any two wells $m^{j(\alpha)}$ and $m^{j(\beta)},\, \alpha, \beta\in\mathcal{J},$ we denote the set of minimal geodesics between them by 
$\mathcal{G}_{j(\alpha), j(\beta)}$ and assume that there are open sets $\Omega^{j(\alpha)}$ and $\Omega^{j(\beta)}$, satisfying Hypothesis 
\ref{hypomega}, such that 
\ben
\item $\mathcal{G}_{j(\alpha), j(\beta)}\subset \Bigl(\Omega^{j(\alpha)}\cup \Omega^{j(\beta)}\Bigr)$.
\item there is a hypersurface $\Sigma_{j(\alpha), j(\beta)}\subset \Bigl(\Omega^{j(\alpha)}\cap \Omega^{j(\beta)}\Bigr)$ transversal to $\mathcal{G}_{j(\alpha), j(\beta)}$.
\item there is a constant $C>0$ such that for all $m\in \Sigma_{j(\alpha), j(\beta)}$ and with the notation $H_{j(\alpha),j(\beta)}:= \mathcal{G}_{j(\alpha), j(\beta)}\cap \Sigma_{j(\alpha), j(\beta)}$
\begin{equation}\label{hyp4_0}
d^{j(\alpha)}(m) + d^{j(\beta)}(m) \geq d(m^{j(\alpha)},m^{j(\beta)}) + 
C \dist (m,H_{j(\alpha),j(\beta)})^2\, .
\end{equation}
\item either $\mathcal{G}_{j(\alpha), j(\beta)}$ consists of a unique minimal 
geodesic (Case I) (in which case we set $H_{j(\alpha),j(\beta)}=: m_0$) or it is a manifold 
(possibly singular at the wells) of dimension 
$\ell+1$ with $1\leq\ell\leq n-1$ (Case II).
\een
\een
\end{hyp}

To unify our notation, we set $\ell=0$ in Case I. Estimate \eqref{hyp4_0} implies that the transverse Hessian of $d^{j(\alpha)} + d^{j(\beta)}$ 
(transverse with respect to $\mathcal{G}_{j(\alpha), j(\beta)}$)
is non-degenerate at all points of $H_{j(\alpha),j(\beta)}$ (the intersection of the geodesics with the hypersurface $\Sigma_{j(\alpha), j(\beta)}$). 
More precisely, choose near $H_{j(\alpha),j(\beta)}$ a 
tubular neighborhood $\tau$ of $\Sigma_{j(\alpha), j(\beta)}$ and commuting unit vector fields $N_1, \ldots N_n$ such that $N=N_n$ is normal to 
$\Sigma_{j(\alpha), j(\beta)}$, $N_1, \ldots N_{n-1}$ are an orthonormal base in $T\Sigma_{j(\alpha), j(\beta)}$ and
$N_{1}, \ldots N_{n-\ell-1}$ are transversal to $\mathcal{G}_{j(\alpha), j(\beta)}$. We remark that $N_n$ is not necessarily tangent to the geodesics in $\mathcal{G}_{j(\alpha), j(\beta)}$ and that the vector fields 
$N_{n-\ell}, \ldots N_{n-1}$ are possibly only locally defined on $H_{j(\alpha),j(\beta)}$ (while $N_1, \ldots N_{n-\ell-1}$ exist globally on $H_{j(\alpha),j(\beta)}$).  Then
\begin{equation}\label{transversHessian} 
 D^2_{\alpha,\beta} := D^2_{\perp, H_{j(\alpha),j(\beta)}}\bigl(d^{j(\alpha)} + d^{j(\beta)}\bigr) := \Bigl( N_s N_t (d^{j(\alpha)} + d^{j(\beta)})|_{H_{j(\alpha),j(\beta)}}\Bigr)_{1\leq s,t \leq n-\ell-1}
\end{equation}
is called the transverse Hessian of $d^{j(\alpha)} + d^{j(\beta)}$ at $H_{j(\alpha),j(\beta)}$ and \eqref{hyp4_0} implies that it is positive 
(in particular non-degenerate) for all points in $H_{j(\alpha),j(\beta)}$. 

Then, using the Morse-Lemma with parameters, the integral in \eqref{prop4_0} (or \eqref{prop4_0a}) for $w_{\alpha\beta}$ has a complete
asymptotic expansion. More precisely,

\begin{theo}\label{theo6}
Under the assumptions given in Hypothesis \ref{hyp4}, for a fixed pair $\alpha, \beta\in\mathcal{J}$, let $w_{\alpha\beta}$ be the interaction matrix element with respect to the spectral interval $I$ as given in \eqref{intact}. 
Then there is a sequence $(I_p)_{p\in \N/2}$ in $\R$ such that
\begin{equation}\label{theo6_0} 
w_{\alpha\beta} \sim \hbar^{-(N_{\alpha} + N_{\beta})} \hbar^{(1-\ell)/2} e^{-S_{j(\alpha),j(\beta)}/\hbar} \sum_{p\in \N/2} \hbar^p I_p \, . 
\end{equation}
Explicit formulae for the leading order term are slightly different in Case I and Case II (see Hypothesis \ref{hyp4}, 3d).

Partition $\mathcal{J}$  into maximal subsets $\mathcal{J}_j$ associated to one potential minimum $m_j$. For $\delta \in \mathcal{J}_j$ and $a_\delta \in C^\infty((0, \hbar_0), \Gamma_c^\infty(M, \Eh))$ given in Theorem \ref{theo4}, let 
\[ \tilde{v}_\delta =  \hbar^{-\frac{n}{4}} e^{-\frac{d^{j}}{\hbar}}  \tilde{a}_\delta\qquad \text{with}\quad \tilde{a}_{\delta} := \sum_{\beta \in \mathcal{J}_j} 
c^j_{\delta,\beta}(\hbar) a_{\beta} \]
be the approximate eigenfunctions and
$C^j(\hbar) = \bigl(c^j_{\alpha,\beta}(\hbar)\bigr)_{\alpha,\beta \in \mathcal{J}_j}$ 
be the unitary  matrix as given in Theorem \ref{theo5}.
We denote by $d\sigma$ the Riemannian surface measure on $H_{j(\alpha),j(\beta})$ induced by the Riemannian metric $g$, by $N=N_n$ the unit
normal vector field on $\Sigma_{j(\alpha), j(\beta)}$ pointing from $m^{j(\alpha)}$ to $m^{j(\beta)}$ and define the transverse Hessian by 
\eqref{transversHessian}.  Moreover, $\dd f (N) = N(f)$ denotes the normal derivative of $f\in\Ce^\infty(M, \R)$.

Then the leading order in the expansion in \eqref{theo6_0} is given by
\ben
\item[Case 1:]
\begin{equation}\label{0thm1} 
I_0 =  (2\pi)^{\frac{n-1}{2}}\left(\det D^{2}_{\alpha,\beta} \right)^{-1/2} 
\dd \bigl(d^{j(\beta)}- d^{j(\alpha)}\bigr)(N)(m_0) \gamma_{m_0} [ \tilde{a}_{\alpha, -N_{\alpha}}, \tilde{a}_{\beta, -N_\beta}]
\end{equation}
\item[Case 2:]
\begin{equation}\label{0thm2}
I_0 =  (2\pi)^{\frac{n-\ell-1}{2}} \int_{H_{j(\alpha),j(\beta)}} \bigl({\det D^{2}_{\alpha,\beta}}\bigr)^{-1/2} 
\dd \bigl( d^{j(\beta)} - d^{j(\alpha)}\bigr)(N)(m) \gamma_{m} [ \tilde{a}_{\alpha, -N_{\alpha}}, \tilde{a}_{\beta, -N_\beta}] \, d\sigma (m)\, .
\end{equation}
\een

\end{theo}

We remark that all $I_p = I_{p;\alpha,\beta}$ depend on $\alpha,\beta$. Moreover, the leading order term satisfies $I_{0;\alpha,\beta} = \overline{I_{0;\beta,\alpha}}$ 
(since $\gamma_m$ is Hermitian and switching $\alpha$ and $\beta$ implies switching the orientation of $N$). 
Thus $I_{0;\alpha,\beta}$ gives the leading order of the (by construction) self-adjoint matrix 
$\Bigl( \tilde{m}_{\alpha\beta} - \mu_\alpha \delta_{\alpha\beta}\Bigr)_{\alpha,\beta\in\mathcal{J}}$ (see Theorem \ref{theo3}) 
describing the interaction between the wells $m^{j(\alpha)}$ and $m^{j(\beta)}$.

Recall that by construction, the eigenvalues of $H_\ep$ exponentially close to $\mu_\alpha$ for $\alpha\in \mathcal{J}$
(given in Hypothesis \ref{hyp4}) also lie in the spectral interval $I$. 
Thus, by Corollary  \ref{cor2} specialized to the case of exactly two elements in $\mathcal{J}$, the operator $H_\ep$ has precisely two
eigenvalues $\lambda_\pm$ inside $I$. Up to errors $\expord{2\sigma}$ (for any $\sigma<S_2$), these are given by the eigenvalues of the $2\times 2$-matrix
$\begin{pmatrix} \mu_\alpha & \tilde{w}_{\alpha\beta}\\ \tilde{w}_{\alpha\beta} & \mu_\beta\end{pmatrix}$, namely
\[ \lambda_{\pm} = \frac{\mu_\alpha + \mu_\beta}{2} \pm \sqrt{\frac{1}{4}(\mu_\alpha - \mu_\beta)^2 + \tilde{w}_{\alpha\beta}^2} + \expord{2\sigma}\; .
\]
Thus in this case the eigenvalue splitting is explicitly given by 
$$\lambda_+ - \lambda_- = 2\sqrt{\frac{1}{4}(\mu_\alpha - \mu_\beta)^2 + \tilde{w}_{\alpha\beta}^2} + \expord{2\sigma}.$$
In the symmetric case with $\mu_\alpha = \mu_\beta$, the splitting is, modulo $\expord{2\sigma}$, given by the symmetric interaction term 
$\tilde{w}_{\alpha\beta} = \frac{1}{2}(w_{\alpha\beta}+w_{\beta\alpha})$.
The complete asymptotic expansion of $\tilde{w}_{\alpha\beta}$ (via expansion of both $w_{\alpha\beta}, w_{\beta\alpha})$) given in 
Theorem \ref{theo6} also gives  an asymptotic expansion of the eigenvalue splitting, if 
$\mu_\alpha - \mu_\beta = \expord{A}$, where $A > S_{j(\alpha),j(\beta)}$.

\begin{proof}
We only prove Case 2.
Theorem \ref{theo5} allows to replace modulo terms of order $O(\hbar^\infty)$ the Dirichlet eigenfunctions $v_\alpha$ and $v_\beta$ in \eqref{prop4_0b} 
or \eqref{prop4_0a} respectively by the 
associated approximate eigenfunctions $\tilde{v}_\alpha$ and $\tilde{v}_{\beta}$.

In fact, for the first term in the integrand on the rhs of equation \eqref{prop4_0b} one obtains
\begin{equation}  \label{comp1}
\gamma_m[\nabla^\Eh_N v_\alpha, v_\beta] - \gamma_m[\nabla^\Eh_N \tilde{v}_\alpha, \tilde{v}_\beta] =
\gamma_m[\nabla^\Eh_N (v_\alpha - \tilde{v}_{\alpha}, v_\beta] +
\gamma_m[\nabla^\Eh_N \tilde{v}_\alpha, v_\beta- \tilde{v}_\beta]
\end{equation}
where $N$ is the unit normal vector field on $\Sigma_{j(\alpha), j(\beta)}$ pointing from $m^{j(\alpha)}$ to $m^{j(\beta)}$.
Writing $\alpha = (j,\ell)$ and $\beta = (k, u)$, and using (for $w$ equal to 
$e^{d^j/\hbar} v_\alpha$ or $e^{d^j/\hbar} ( v_\alpha - \tilde{v}_\alpha)$)
the identity
\begin{equation} \label{comp2}
\nabla^\Eh_N e^{-d^j/\hbar} w=  e^{-d^j/\hbar} \left( -\frac{1}{\hbar} \dd d^j(N)\otimes w + \nabla^\Eh_N w \right),
\end{equation}
straightforward calculation
of $\int_{\Sigma_{jk}} \text{rhs} \eqref{comp1}\, d\sigma$ gives, by use of the estimate \eqref{surface}, Schwarz inequality and Corollary \ref{cor1}, an error of order $O(\hbar^\infty e^{-S_{j,k}/\hbar} )$. 
Treating the second  term in the integrand on the right hand side of equation \eqref{prop4_0b}  in the same way proves our claim. 

Thus, modulo $O\Bigl(\hbar^{\infty}e^{-S_{j,k}/\hbar}\Bigr)$, we  
get, using $S_{j,k} < S_0 +a$ in the case specified in \eqref{prop4_00} (or $S_{j,k}=S_0$ in the case
specified in \eqref{prop4_00a}, leading to the estimate \eqref{prop4_0a}), the representation formula

\begin{align*}
w_{\alpha\beta}& \equiv \hbar^2 \int_{\Sigma_{j,k}} \Bigl( \gamma_m[\nabla^\Eh_N \tilde{v}_\alpha, \tilde{v}_\beta] - 
\gamma_m[\tilde{v}_\alpha, \nabla^\Eh_N \tilde{v}_\beta] \Bigr)\, d\sigma (m) 
\\
&= \hbar^{2-\frac{n}{2}} \int_{\Sigma_{j,k}}  \Bigl( \gamma_m[\nabla^\Eh_N e^{-\frac{d^j}{\hbar}}\tilde{a}_\alpha(\cdot , \hbar),  e^{-\frac{d^k}{\hbar}} \tilde{a}_\beta(\cdot , \hbar)]  - 
\gamma_m[ e^{-\frac{d^j}{\hbar}}\tilde{a}_\alpha(\cdot , \hbar), \nabla^\Eh_N e^{-\frac{d^k}{\hbar}}\tilde{a}_\beta(\cdot , \hbar)] \Bigr)\, d\sigma (m) 
 \label{theo6_1}
   \end{align*}
where in the second equation we used \eqref{hatvjk}, \eqref{theo5_0} and the definition of $\tilde{a}$.
Using 
\[ \nabla^\Eh_N e^{-\frac{d^k}{\hbar}}\tilde{a}_\beta(\cdot , \hbar) = e^{-\frac{d^k}{\hbar}}\bigl(\frac{1}{\hbar} \dd d^k(N) \tilde{a}_\beta (\cdot , \hbar) + \nabla^\Eh_N \tilde{a}_\beta (\cdot , \hbar)\bigr) \]
and the notations 
\begin{align} \label{theo6_11}
\varphi_{jk} &:= d^j + d^k - S_{j,k}\\
\eta_{\alpha ,\beta}(m, \hbar) &:=  -\gamma_m [\dd d^j(N) \tilde{a}_\alpha(\cdot , \hbar),  \tilde{a}_\beta(\cdot , \hbar)] +  
 \gamma_m[\tilde{a}_\alpha(\cdot , \hbar),  \dd d^k (N) \tilde{a}_\beta(\cdot , \hbar)] \nonumber \\
 \mu_{\alpha ,\beta}(m, \hbar) &:= \gamma_m [ \nabla^\Eh_N \tilde{a}_\alpha(\cdot , \hbar),  \tilde{a}_\beta(\cdot , \hbar)] - \gamma_m[ \tilde{a}_\alpha(\cdot , \hbar),  \nabla^\Eh_N \tilde{a}_\beta(\cdot , \hbar)] \nonumber
\end{align}
we therefore get  
\begin{equation}\label{theo6_2}
 w_{\alpha \beta} \equiv \hbar^{2-\frac{n}{2}} e^{-\frac{S_{jk}}{\hbar}}   \int_{\Sigma_{j,k}}e^{-\frac{\varphi_{jk}}{\hbar} } \Big[ \frac{1}{\hbar} \eta_{\alpha,\beta}(m,\hbar) +  \mu_{\alpha,\beta}(m,\hbar) \Bigr] \, d\sigma (m) 
\, .
\end{equation}

We now use an adapted version of stationary phase. We choose vector fields on $H_{jk}$ as described above equation \eqref{transversHessian} and 
a (sufficiently small) tubular neighborhood $\tau$ of $H_{jk}$ on $\Sigma_{jk}$. 

Using the Tubular Neighborhood Theorem, there is a diffeomorphism 
\begin{equation}\label{theo6_3} 
\kappa: \tau \rightarrow H_{jk}\times (-\delta, \delta)^{n-\ell-1} \, , \quad k(x) = (s,t)
\end{equation}
 such that
for each $x\in \tau$ there exists exactly one $s\in H_{jk}$ and
$t\in (-\delta, \delta)^{n-\ell-1}$ such that 
\begin{equation}\label{theo6_4}
x= s + \sum_{m=1}^{n-\ell-1} t_{m} N_m(s) \quad\text{for}\quad \kappa(x) = (s, t) \, . 
\end{equation}
This follows from the proof of the Tubular Neighborhood Theorem, see e.g. \cite{hirsch}. It allows to continue the vector fields $N_m,\, m=1, \ldots n-\ell-1,$ from 
$H_{jk}$ to $\tau$ by setting $N_m(x):= N_m(s)$, thus
$N_m = \partial_{t_{m}}$ and the vectorfields $N_m$ commute.

For $\varphi_{jk}$ given in \eqref{theo6_11}, we define 
\[ \tilde{\varphi}_{jk}:= \varphi_{jk} \circ \kappa^{-1} : H_{jk}\times (-\delta, \delta)^{n-\ell-1} \rightarrow \R \quad\text{ with }\quad
\tilde{\varphi}_{jk}(s,t) := \varphi_{jk} \circ \kappa^{-1}(s,t) = \varphi_{jk} (x)\; .\]
Then $H_{jk}$ is given by $t=0$ and
\begin{equation}\label{theo6_5}
 \tilde{\varphi}_{jk}|_{H_{jk}} = \varphi_{jk}|_{H_{jk}} = 0 \quad\text{and} \quad \varphi_{jk}(x)>0\; \text{ for }\; x\in \tau\setminus H_{jk}\, .
 \end{equation}
Moreover, since $\varphi_{jk}$ is constant and minimal on $H_{jk}$ and since $\varphi_{jk}(x)\geq \dist (x, H_{jk})^2$ by equation \eqref{hyp4_0} in Hypothesis \ref{hyp4},
\begin{align}
 N_m \varphi_{jk}|_{H_{jk}} &= \frac{\partial}{\partial s_m} \tilde{\varphi}_{jk}|_{H_{jk}} = 0 \quad\text{for}\quad m=n-\ell, \ldots n-1 \nonumber\\
 N_m \varphi_{jk}|_{H_{jk}} &= \frac{\partial}{\partial t_m} \tilde{\varphi}_{jk}|_{H_{jk}} = 0 \quad\text{for}\quad m=1, \ldots n-\ell-1 \nonumber \\
 \Bigl( N_mN_u \varphi_{jk}|_{H_{jk}}\Bigr)_{1\leq m,u \leq n-\ell-1} & = D^2_t \tilde{\varphi}_{jk}|_{t=0} > 0\; .\label{theo6_6}
\end{align}
Now we use the following adapted version of the Morse-Lemma with parameters (see e.g. \cite{Dui}, \cite{lang}).

\begin{Lem}\label{Morse}
	Let $\phi\in\Ce^\infty \bigl(H_{jk}\times (-\delta, \delta)^{n-\ell-1}\bigr)$ be such that $\phi(s,0) = 0$, $D_t\phi (s,0) = 0$ and the
	transversal Hessian $D^2_t\phi(s,\cdot)|_{t=0} =: Q(s)$ is non-degenerate for all $s\in H_{jk}$. Then, for each $s\in H_{jk}$, there is a diffeomorphism
	$ y(s,.):  (-\delta, \delta)^{n-\ell-1} \rightarrow U$, where $U \subset\R^{n-\ell-1}$ is some neighborhood of $0$, such that 
	\begin{equation}\label{3.2thm2} 
	y(s,t) = t + O\bigl(|t|^2\bigr) \quad\text{as}\;\; |t|\to 0\quad\text{and}\quad \phi(s,t) = \frac{1}{2}\langle y(s,t),  Q(s) y(s,t)\rangle\; . 
	\end{equation}
	Furthermore, $y(s,t)$ is $\Ce^\infty$ in $s\in H_{jk}$.
\end{Lem}

By \eqref{theo6_5} and \eqref{theo6_6}, the phase function $\tilde{\varphi}_{jk}$ satisfies the assumptions on $\phi$ given in Lemma \ref{Morse}.
We thus can define the diffeomorphism $h:= \id \times y: H_{jk}\times (-\delta, \delta)^{n-\ell-1}\rightarrow H_{jk} \times U$ for
$y$ constructed with respect to $\tilde{\varphi}_{jk}$ as in Lemma \ref{Morse}. 
Using the diffeomorphism $\kappa$ as given in \eqref{theo6_3} and \eqref{theo6_4}, 
we set $g(x) := h\circ \kappa (x) = (s, y)$ (then $g^{-1}(s,0) = s$ for $s\in H_{jk}$). Thus
\begin{equation}\label{theo6_7}
\varphi_{jk} \bigl(g^{-1}(s,y)\bigr) = \frac{1}{2} \langle  y, Q(s) y\rangle
\end{equation}
and, modulo $O\bigl(\hbar^\infty e^{-\frac{S_{jk}}{\hbar}}  \bigr)$, we obtain by \eqref{theo6_2}, setting $x:=g^{-1}(s,y)$,
\begin{equation} \label{theo6_8}
w_{\alpha\beta} \equiv \hbar^{2-\frac{n}{2}} e^{-S_{jk}/\hbar} \int_{H_{j,k}} \int_U e^{-\langle  y, Q(s) y\rangle/2\hbar} 
\Bigl[ \frac{1}{\hbar}\eta_{\alpha,\beta}(g^{-1}(s,y)) + \mu_{\alpha,\beta}(g^{-1}(s,y)) \Bigr] J(s,y) \, dy\, d\tilde{\sigma} (s)
\end{equation}
where $d\tilde{\sigma}$ is the Euclidean surface measure on $H_{jk}$ and $J(s,y)= \det D_y g^{-1}(s, y)$ denotes the Jacobi determinant  for the
diffeomorphism $g^{-1}(s, .)$ which maps $U$ into a subset of $\Span \bigl(N_{1}(s), \- \ldots, N_{n-\ell-1}(s)\bigr)$
and $Q(s)=D^2_t\tilde{\varphi}_{jk}(s,\cdot)|_{t=0}$ denotes the transversal Hessian of $\tilde{\varphi}_{jk}$ 
as given in \eqref{theo6_6}.  From the construction of $g$ and \eqref{theo6_4} it follows that $J(s,0) = 1$ for all $s\in H_{jk}$.

By the stationary phase formula with respect to $y$ in \eqref{theo6_8}, we get modulo $O\bigl(e^{-\frac{S_{jk}}{\hbar}} \hbar^\infty \bigr)$
\begin{align}\label{theo6_9}
w_{\alpha\beta} 
&\equiv  \hbar^{\frac{1- \ell}{2}} e^{-\frac{S_{jk}}{\hbar}}\bigl(2 \pi\bigr)^{\frac{n-\ell-1}{2}}\sum_{\nu=0}^\infty
\hbar^\nu  \int_{H_{j,k}}B_{\nu} (s)\, d\sigma(s) \quad\text{where}\\
B_{\nu} (s) &= \bigl(\det Q(s)\bigr)^{-\frac{1}{2}} \frac{1}{\nu !}\Bigl( \langle \partial_y, Q^{-1}(s) \partial_y\rangle^\nu 
J \bigl( \eta_{\alpha,\beta}+ \hbar \mu_{\alpha,\beta}\bigr)\circ g^{-1} \Bigr)(s,0) \, .\nonumber
\end{align}
In particular, for any $s\in H_{j,k}$, using the notation \eqref{transversHessian}, $B_{0}(s)$ is given by the leading order of
\begin{equation}\label{theo6_10}
\Bigl(\det Q(s)\Bigr)^{-\frac{1}{2}} \eta_{\alpha,\beta}\circ g^{-1}(s,0) 
=  \Bigl|\det D^2_{\alpha\beta}(s)\Bigr|^{-\frac{1}{2}} \eta_{\alpha,\beta}(s)\, ,
\end{equation}
using \eqref{theo6_6} and identifying $s\in H_{j,k}$ with a point in $M$.

We now use the definition of $\eta_{\alpha,\beta}$ in equation \eqref{theo6_11} and the expansions of $\tilde{a}_\alpha \tilde{a}_\beta$ as given in Theorem \ref{theo4} to get
\begin{align}
\eta_{\alpha,\beta}(x) &= \Bigl(\dd d^k(N) - \dd d^j (N)\Bigr)(x) \gamma_x \bigl[\tilde{a}_\alpha(\cdot , \hbar), \tilde{a}_\beta(\cdot , \hbar)\bigr] \nonumber\\
&= \hbar^{-(N_\alpha+N_\beta)}\Bigl(\dd d^k(N) - \dd d^j (N)\Bigr)(x) \gamma_x \bigl[\tilde{a}_{\alpha,N_\alpha},\tilde{a}_{\beta, N_\beta}\bigr] + O\Bigl(\hbar^{-(N_\alpha+N_\beta)+ \frac{1}{2}}\Bigr) \label{8.2thm2}
\end{align}
Combining \eqref{8.2thm2}, \eqref{theo6_10} and \eqref{theo6_9} completes the proof.
\end{proof}

\end{document}